%% file: arxiv.tex
\documentclass[11pt]{article}
\usepackage[papersize={8.5in,11in},top=1in,left=1in,right=1in,bottom=1in]{geometry}
\usepackage{amsmath, amsthm}
\usepackage{amssymb}
\usepackage{graphicx}
\usepackage{epsfig,xcolor}
\usepackage{cite}

\newtheorem{theorem}{Theorem}
\newtheorem{lemma}{Lemma}
\newtheorem{corollary}{Corollary}

\newtheorem{definition}{Definition}

\newcommand{\sm}{SM(G,\vec{\alpha})}
\newcommand{\ep}{\epsilon}

\title{
  Friendship, Altruism,
  and Reward Sharing in\\
  Stable Matching and Contribution Games\thanks{This work was supported in part by NSF grants CCF-0914782 and CCF-1101495.}}

\author{%
  Elliot Anshelevich\thanks{Dept. of Computer Science, Rensselaer
    Polytechnic Institute, Troy, NY}
  \and Onkar Bhardwaj\thanks{Dept. of Electrical, Computer, \& Systems
    Engineering, Rensselaer Polytechnic Institute, Troy, NY}
  \and Martin Hoefer\thanks{Dept. of Computer Science, RWTH Aachen
    University, Germany}
}

\begin{document}
\maketitle
\thispagestyle{empty}

\begin{abstract}
  We study stable matching problems in networks where players
  are embedded in a social context, and may incorporate friendship
  relations or altruism into their decisions. Each player is a node in
  a social network and strives to form a good match with a neighboring
  player. We consider the existence, computation, and inefficiency of stable
  matchings from which no pair of players wants to deviate. When the benefits
  from a match are the same for both players, we show that incorporating
  the well-being of other players into their matching decisions significantly
  decreases the price of stability, while the price of anarchy remains unaffected.
  Furthermore, a good stable matching achieving the price of stability bound always
  exists and can be reached in polynomial time. We extend these results to more
  general matching rewards, when players matched to each other may receive different
  utilities from the match. For this more general case, we show that incorporating
  social context (i.e., ``caring about your friends") can make an even larger
  difference, and greatly reduce the price of anarchy. We show a variety of existence
  results, and present upper and lower bounds on the prices of anarchy and
  stability for various matching utility structures. Finally, we extend most of our
  results to network contribution games, in which players can decide how much effort
  to contribute to each incident edge, instead of simply choosing a single node to match with.
\end{abstract}

\setcounter{page}{1}

\input{introduction}
\input{fship-equal-sharing}

\input{reward-sharing}
\input{fship-reward-sharing}
\input{ccg-1}

\input{ccg-2}
\bibliographystyle{plain}

\input{arxiv.bbl}
\end{document}

%% file: introduction.tex
\section{Introduction}

Stable matching problems form the basis of many important assignment
and allocation tasks in economics and computer science. The central
approach to analyzing such scenarios is two-sided matching, which has
been studied intensively since the 1970s in both the algorithms and
economics literature~\cite{Gusfield89,RothBook90}. An important
variant of stable matching is matching with cardinal utilities, when
each match can be given numerical values expressing the \emph{quality}
or \emph{reward} that the match yields for each of the incident
players~\cite{Anshelevich10}. Cardinal utilities specify the quality
of each match instead of just a preference ordering, and they allow
the comparison of different matchings using measures such as social
welfare. A particularly appealing special case of cardinal utilities
is known as correlated stable matching, where both players who are
matched together obtain the same reward. Apart from the wide-spread
applications of correlated stable matching in, e.g., market
sharing~\cite{Goemans06}, job markets~\cite{Arcaute09}, social
networks~\cite{HoeferICALP11}, and distributed computer
networks~\cite{Goemans06,Mathieu08}, this model also has favorable
theoretical properties such as the existence of a potential function. It
guarantees existence of a stable matching even in the non-bipartite
case, where every pair of players is allowed to
match~\cite{Abraham08,Mathieu08}.

When matching individuals in a social environment, it is often
unreasonable to assume that each player cares only about their own
match quality. Instead, players may incorporate the well-being of
their friends/neighbors as well, or that of
friends-of-friends. Players may even be altruistic to some degree, and
consider the welfare of all players in the network. Caring about
friends and altruistic behavior is commonly observed in practice and
has been documented in laboratory
experiments~\cite{Levine98,Eshel98}. However, results in algorithmic
game theory about the impact of social context on stable outcomes are
only recently starting to
appear~\cite{Ashlagi08,HoeferIJCAI11,Buehler11,Chen08,HoeferESA09,HoeferCoor09,Chen11}.
In this paper, we study how social context influences stability and efficiency in matching scenarios. We use a general approach incorporating the social
context of a player into its decisions. Every player may consider the
well-being of every other player to some degree, with the degree of
this regardfulness possibly decaying with the hop distance in the
network. Players who only care about their neighbors, as well as fully
altruistic players, are special cases of this model. Our model of
altruism is a strict generalization of recent approaches in
algorithmic game theory in which the social welfare of the whole
population is (part of) the utility of each player. 

Moreover, for matching in social environments, the standard model of
correlated stable matching may be too constraining compared to general
cardinal utilities, because matched players receive exactly the same
reward. Such an \emph{equal sharing} property is intuitive and bears a
simple beauty, but there are a variety of other reward sharing methods
that can be more natural in different contexts. For instance, in
theoretical computer science it is common practice to list authors
alphabetically, but in other disciplines the author sequence is
carefully designed to ensure a proper allocation of credit to the
different participants of a joint paper. Here the credit is often
supposed to be allocated in terms of input, i.e., the first author
should be the one that has contributed most to the project. Such
input-based or proportional sharing is then sometimes overruled with
sharing based on intrinsic or acquired social status, e.g., when a
distinguished expert in a field is easily recognized and
subconsciously credited most with authorship of an article. In this
paper, we are interested in how such unequal reward sharing rules affect
stable matching scenarios. In particular, we consider a large class of
local reward sharing rules and characterize the impact of unequal
sharing on existence and inefficiency of stable matchings, both in
cases when players are embedded in a social context and when they are not.

Recently, correlated matching problems have become the basis for
analyzing more general contribution and participation games in
networks. In such games, each player must decide how much effort to
contribute to each relationship or project that it is involved
in. Insights on such problems may advance the understanding of
contribution incentives in networked societies and improve the design
of user-based platforms. As we show in Sections~\ref{sec:CCG}
and~\ref{sec:CCG2}, we are able to extend most of our results about
stable matching in the presence of social context and general reward
sharing to \emph{network contribution games} which have recently been
introduced in~\cite{AnshelevichAlgo11}.

\subsection{Stable Matching and Contribution Games}
\label{sec:model}

In this paper we consider two classes of games: stable matching with
cardinal utilities, and convex contribution games. We consider both
scenarios in the presence of social context, and unequal reward
sharing.

\paragraph{Stable Matching} Correlated stable matching is a prominent
subclass of general ordinary stable matching. In this game, we are
given a (non-bipartite) graph $G=(V,E)$ with edge weights $r_e$. In a
matching $M$, if node $u$ is matched to node $v$, the utility of node
$u$ is defined to be exactly $r_e$. This can be interpreted as both
$u$ and $v$ getting an identical reward from being matched
together. We will also consider unequal reward sharing, where node $u$
obtains some reward $r_e^u$ and node $v$ obtains reward $r_e^v$ with
$r_e^u+r_e^v=r_e$. Therefore, the preference ordering of each node
over its possible matches is implied by the rewards that this node
obtains from different edges. A pair of nodes $(u,v)$ is called a {\em
  blocking pair} in matching $M$ if $u$ and $v$ are not matched to
each other in $M$, but can both strictly increase their rewards by
being matched to each other instead. A matching with no blocking pairs
is called a {\em stable matching.}

While the matching model above has been well-studied, in this paper we
are interested in stable matchings that arise in the presence of
social context. Denote the reward obtained by a node $v$ in a matching
$M$ as $R_v$. We now consider the case when node $u$ not only cares
about its own reward, but also about the rewards of its
friends. Specifically, the {\em perceived} or {\em friendship utility}
of node $v$ in matching $M$ is defined as
$$U_v=R_v+\sum_{d=1}^{diam(G)}\alpha_d\sum_{u\in N_d(v)}R_u,$$
where $N_d(v)$ is the set of nodes with shortest distance exactly $d$
from $v$, and $1\geq\alpha_1\geq\alpha_2\geq\ldots\geq 0$ (we use
$\vec{\alpha}$ to denote the vector of $\alpha_i$ values). In other
words, for a node $u$ that is distance $d$ away from $v$, the utility
of $v$ increases by an $\alpha_d$ factor of the reward received by
$u$. Thus, if $\alpha_d=0$ for all $d\geq 2$, this means that nodes
only care about their neighbors, while if all $\alpha_d>0$, this means
that nodes are altruistic and care about the rewards of everyone in
the graph. The perceived utility is the quantity that the nodes are
trying to maximize, and thus, in the presence of friendship, a
blocking pair would be a pair of nodes such that each could increase
its {\em perceived utility} by matching to each other.

\paragraph{Contribution Games} While most of the results in this paper
concern stable matching, we also study convex contribution games (CCG)
(for detailed definition and discussion see
\cite{AnshelevichAlgo11}). In these games, we are given a graph
$G=(V,E)$ with a {\em reward function} $f_e:\mathbb{R}_{\geq
  0}\times\mathbb{R}_{\geq 0}\rightarrow\mathbb{R}_{\geq 0}$ for each
edge $e$, which is assumed to be nondecreasing and convex in each of
its arguments, and obeys the property that $f_e(0,y)=f_e(x,0)=0$ for
all $x,y$. The nodes are players of this game: each node $v$ has a
budget $B_v$, and its strategy consists of deciding how to allocate
this budget among its incident edges. The reward to node $v$ from edge
$e=(v,u)$ is equal to $f_e(s_v(e),s_u(e))$, where $s_v(e)$ and
$s_u(e)$ are the amounts of budget allocated to edge $e$ by nodes $v$
and $u$ respectively. For the case where reward to endpoints of an
edge is allowed to be different, we instead have two functions:
$f_e^v$ and $f_e^u$ such that $v$ receives $f_e^v(s_v(e),s_u(e))$
reward and $u$ receives $f_e^u(s_v(e),s_u(e))$ reward.  The total
reward of a node $v$ (which we denote by $R_v$) is simply the total
reward it collects from incident edges.

In this paper, just as in \cite{AnshelevichAlgo11}, we will mostly be
concerned with {\em pairwise equilibria} of CCG. A pairwise
equilibrium (a.k.a. a 2-strong equilibrium) is a solution where no
pair of players can switch their strategies (budget allocations)
simultaneously such that both players strictly increase their
rewards. Also, a pairwise equilibrium must not possess any unilateral
improving deviations by any individual player.

Just as with stable matching, we are interested in the properties of
CCG with social context and friendship utilites. This version is
defined analogously: the perceived utility of a node $v$ is
$U_v=R_v+\sum_{d=1}^{diam(G)}\alpha_d\sum_{u\in N_d(v)}R_u,$ and this
is what node $v$ is attempting to maximize. Therefore, a pairwise
equilibrium in the presence of friendship means that there is no pair
of nodes that can simultaneously improve their perceived utility.

\paragraph{Centralized Optimum and the Price of Anarchy} We study the
social welfare of equilibrium solutions and compare them to an optimal
centralized solution. The social welfare is the sum of rewards, i.e.,
the optimal solution is the one that maximizes $\sum_v R_v$. Notice
that, while this is equivalent to maximizing the sum of player
utilities when $\vec{\alpha}=0$, this is no longer true with social
context (i.e., when $\vec{\alpha}\neq 0$). Nevertheless, as in
e.g.~\cite{Meier08,Chen11}, we believe this is a well-motivated and
important measure of solution quality, as it captures the overall
performance of the system, while ignoring the perceived ``good-will"
effects of friendship and altruism. For example, when considering
projects done in pairs, the reward of an edge can represent actual
productivity, while the perceived utility may not.

To compare stable solutions with the centralized optimum, we will
often consider the price of anarchy and the price of stability. When
considering stable matchings, by the price of anarchy
(resp. stability) we will mean the ratio between the optimum
centralized solution and the worst (resp. best) stable
matching. Similarly, when considering CCG, by the price of anarchy
(resp. stability) we will mean the ratio between the optimum
centralized solution and the worst (resp. best) {\em pairwise
  equilibrium}.

\subsection{Our Results}

For stable matching with cardinal utilities we show the following.

\begin{itemize}
\item For friendship utilities and equal reward sharing, a stable
  matching exists and the price of anarchy (ratio of the
  maximum-weight matching with the worst stable matching) is at most
  2, the same as in the case without friendship. The price of
  stability, on the other hand, improves in the presence of
  friendship, as we can show a tight bound of
  $\frac{2+2\alpha_1}{1+2\alpha_1+\alpha_2}$. Moreover, we present a
  dynamic process that converges to a stable matching of at least this
  quality in polynomial time, if initiated from the maximum-weight
  matching.

\item When two nodes matched together may receive different rewards, a
  stable matching may not exist. However, for several natural local
  reward sharing rules (e.g., when reward shares depend on inherent
  properties of the two incident nodes, see
  Section~\ref{sec:reward-sharing}), we show that a stable matching
  exists. Moreover, for arbitrary oblivious reward sharing (i.e., when rewards for the incident players are arbitrary but independent of the matching decisions of other players), we show that prices of anarchy and stability depend on the level of inequality among reward shares. Specifically, if $R$ is the maximum ratio over all edges $(u,v) \in E$ of the reward shares of node $u$ and $v$, then the price of anarchy is at most $1+R$ without friendship, and at most $\frac{(1+R)(1+\alpha_1)}{1+\alpha_1R}$ with
  friendship utilities. We also show tight or almost-tight lower bounds on
  the price of anarchy, and give improved results for several
  particular reward sharing rules.
\end{itemize}

Our results imply that for socially aware players, the price of
stability can greatly improve: e.g., if
$\alpha_1=\alpha_2=\frac{1}{2}$, then the price of stability is at
most $\frac{6}{5}$, and a solution of this quality can be obtained
efficiently. Moreover, if reward sharing is extremely unfair ($R$ is
unbounded), then friendship becomes even more important: changing
$\alpha_1$ from $0$ to $\frac{1}{2}$ reduces the price of anarchy from
being unbounded to being at most $3$.

We next consider convex contribution games. While friendship changes
the properties of these games (e.g., there might be instances without
a strong equilibrium), we show that all of the results mentioned above
for stable matching also hold for convex contribution games, replacing
``stable matching" with ``pairwise equilibrium". For the case where
players do not have to spend all of their budget, this is not
difficult to show, as there is a one-to-one correspondence between
stable matchings and pairwise equilibria. For the case where players
{\em must} spend their entire budget, however, this becomes somewhat
trickier, as the types of deviations available to players in convex
contribution games are significantly more numerous than in stable
matching models. Nevertheless, we show that the same results hold for
the case of {\em local} friendship, i.e., where $\alpha_i=0$ for all
$i\geq 2$. We also show new results for specific reward sharing rules
in convex contribution games, such as proportional sharing, where a
node's share of reward from an edge is proportional to the amount of
effort it contributes to the edge.


\subsection{Related Work}

Stable matching problems have been studied intensively over the last
few decades. On the algorithmic side, existence, efficient algorithms,
and improvement dynamics for two-sided stable matchings have been of
interest (for references, see standard
textbooks~\cite{Gusfield89,RothBook90}). In this paper we address the
more general stable roommates problem, in which every player can be
matched to every other player. For general preference lists, there
have been numerous works characterizing and algorithmically deciding
existence of stable
matchings~\cite{Irving85,Teo98,Chung00,Roth90}. For the correlated
stable roommates problem, existence is guaranteed by a potential
function argument~\cite{Abraham08,Mathieu08}, and convergence time of
random improvement dynamics is polynomial~\cite{Ackermann11}. In
\cite{AnshelevichD09}, price of anarchy and stability bounds for {\em
  approximate} correlated stable matchings were provided. In contrast,
we study friendship, altruism, and unequal reward sharing in stable
roommate problems with cardinal utilities.

Another line of research closely connected to some of our results
involves game-theoretic models for contribution. A prominent example
is the general approach by Ballester et al~\cite{Ballester06}, in
which equilibria exhibit similarities with a commonly known centrality
index in social networks. There are numerous extensions and variants
of this game. In all these games, however, players contribute quite
generally to the whole society, and not to particular links or
relationships. See~\cite{Galeotti10} for an analysis of a broad
framework that includes this game and several others (such as public
goods games~\cite{BramoulleJET07}). Instead,
in~\cite{AnshelevichAlgo11} we consider a contribution game tied more
closely to matching problems. Here players have a budget of effort and
contribute parts of this effort towards specific projects and
relationships. For more related work on the contribution game,
see~\cite{AnshelevichAlgo11}. All previous results for this model
concern equal sharing and do not address the impact of the player's
social context.

Analytical aspects of reward sharing have been a central theme in game
theory since its beginning, especially in cooperative
games~\cite{Peleg03}. Recently, there have been prominent algorithmic
results also for network bargaining~\cite{KleinbergSTOC08,Kanoria11}
and credit allocation problems~\cite{Kleinberg11}. 
In addition, the work in \cite{Edith} considers various reward sharing schemes in coalition formation; their motivation resembles ours, although they mostly consider Nash equilibrium solutions in hypergraphs, while we consider pairwise equilibria in the presence of social context. Work such as \cite{Zick} could also be considered a generalization of contribution games, but in the cooperative setting and without the players having a social context.

Our notion of a player's social context is based on numerical
influence parameters that determine the impact of player rewards on
the (perceived) utilities of other players. A recently popular model
of altruism is inspired by Ledyard~\cite{Ledyard97} and has generated
much interest in algorithmic game
theory~\cite{Chen08,Chen11,HoeferESA09,HoeferCoor09}. Our model
smoothly interpolates between this global approach and the idea of
\emph{surplus collaboration} among players in a given social context
put forward in~\cite{Meier08,Ashlagi08} and considered recently
in~\cite{Buehler11}.

In addition, our work is more generally related to the area of
strategic network creation games, in which selfish players build
networks and optimize different trade-offs between creation cost and
benefits from network structure. For an introduction to this
literature see recent expositions~\cite{TardosChapter07,Jackson08}. In
this literature, there also originated a notion of pairwise
equilibrium, which allows fewer player deviations than what we term
pairwise equilibrium here. In our case, it corresponds exactly to
2-strong equilibrium; for a discussion see~\cite{AnshelevichAlgo11}.

%% file: fship-equal-sharing.tex
\section{Stable Matching with Friendship Utilities}
\label{sec:friendship-equal-sharing}

\newcommand{\BRBP}{{\sc Best-Relaxed-Blocking-Pair}}
\newcommand{\BBP}{{\sc Best-Blocking-Pair}}

We begin by considering correlated stable matching in the presence of
friendship utilities. In this section, the reward received by both
nodes of an edge in a matching is the same, i.e., we use equal reward
sharing, where every edge $e$ has an inherent value $r_e$ and both
endpoints receive this value if edge $e$ is in the matching. We
consider more general reward sharing schemes in
Sections~\ref{sec:reward-sharing} and~\ref{sec:friendship-RS}. Recall
that the friendship utility of a node $v$ increase by $\alpha_d R_u$
for every node $u$, where $d$ is the shortest distance between $v$ and
$u$. We abuse notation slightly, and let $\alpha_{uv}$ denote
$\alpha_d$, so if $u$ and $v$ are neighbors, then
$\alpha_{uv}=\alpha_1$.

%


\begin{figure}
  \centering
  \includegraphics[scale=0.5]{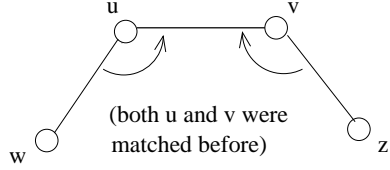}
  \caption{biswivel deviation}
  \label{fig:biswivel}
\end{figure}

\begin{figure}
  \centering
  \includegraphics[scale=0.5]{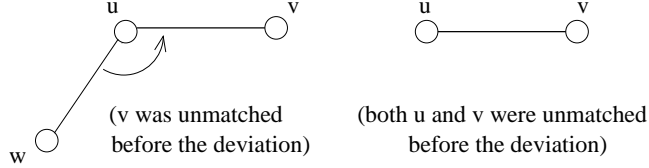}
  \caption{swivel deviation}
  \label{fig:swivel}
\end{figure}

Given a matching $M$, we begin by classifying the following types of
improving deviations that a blocking pair can undergo.

\begin{definition}
  We call an improving deviation a {\bf biswivel} whenever two
  neighbors $u$ and $v$ switch to match to each other, such that both
  $u$ and $v$ were matched to some other nodes before the deviation in
  $M$.
\end{definition}
See Figure~\ref{fig:biswivel} for explanation. For such a biswivel to
exist in a matching, the following necessary and sufficient conditions
must hold.
\begin{eqnarray}
 (1+\alpha_1)r_{uv} &>& (1+\alpha_1)r_{uw}
			+ (\alpha_1+\alpha_{uz})\,r_{vz}
	\label{eqn:biswivel-1}\\
 (1+\alpha_1)r_{uv} &>& (1+\alpha_1)r_{vz}
		      + (\alpha_1+\alpha_{vw})\,r_{uw}
	\label{eqn:biswivel-2}
\end{eqnarray}
Inequality~\eqref{eqn:biswivel-1} can be explained as follows: The
left side quantifies the utility gained by $u$ because of getting
matched to $v$ and the right side quantifies the utility lost by $u$
because of $u$ and $v$ breaking their present matchings with $w$ and
$z$ respectively. Hence Inequality~\eqref{eqn:biswivel-1} implies that
$u$ gains more utility by getting matched with $v$ than it loses
because of $u$ and $v$ breaking their matchings with $v$ and $z$.
Inequality~\eqref{eqn:biswivel-2} can similarly be explained in the
context of node
$v$.
%

\begin{definition}
  We call an improving deviation a {\bf swivel} whenever two neighbors
  get matched such that at least one node among the two neighbors was
  not matched before the deviation.
\end{definition}
See Figure~\ref{fig:swivel} for explanation. For such a swivel to
occur, the following set of conditions must hold.
\begin{eqnarray}
  (1+\alpha_1) r_{uv} &>& (1+\alpha_1) r_{uw}
				\label{eqn:support-401}\\
  (1+\alpha_1) r_{uv} &>& (\alpha_1+\alpha_{vw}) r_{uw}
				\label{eqn:support-403}
\end{eqnarray}
Inequality~\eqref{eqn:support-401} says that $u$ gains more utility by
getting matched with $v$ than it loses by breaking its matching with
$w$. Inequality~\eqref{eqn:support-403} says that $v$ gains more
utility by getting matched with $v$ than the utility it loses because
of $u$ breaking its matching with $w$.  As $\alpha_1+\alpha_{vw} \leq
1+\alpha_1$, Inequality~\eqref{eqn:support-403} is implied by
Inequality~\eqref{eqn:support-401}. This means that if $v$ is
unmatched, the only condition for $(u,v)$ to be a blocking pair is
that $u$ should have net increase in utility by getting matched with
$v$. This is true even if $v$ and $w$ are neighbors. Canceling the
factor of $1+\alpha_1$, we can thus summarize this (necessary and
sufficient) condition for swivel to be an improving deviation as:
\begin{eqnarray}
  r_{uv} > r_{uw} \label{eqn:swivel}
\end{eqnarray}

All improving deviations by a blocking pair can be classified as
either a biswivel or a swivel, depending only on whether both nodes
are matched or not.
Now we make the following observation which will be useful later:

\begin{lemma} \label{lem:increase-edge-reward}
  Suppose a node $u$ is matched to $w$ in matching $M$. If $(u,v)$
  form a blocking pair, then $r_{uv} > r_{uw}$.
\end{lemma}

\begin{proof}
  It is straightforward to see it from
  inequalities~\eqref{eqn:biswivel-1} and~\eqref{eqn:biswivel-2} for a
  biswivel and inequality~\eqref{eqn:swivel} in case of a swivel.
\qed \end{proof}

\subsection{Existence and Price of Anarchy of Stable Matching with Friendship Utilities}
\label{subsec:existence-stable-matching}

\begin{theorem} \label{thm:existence-stable-matching}
  A stable matching exists in stable matching games with friendship
  utilities. Moreover, the set of stable matchings without friendship
  (i.e., when $\vec{\alpha}=\mathbf{0}$) is a subset of the set of
  stable matchings with friendship utilities on the same graph.
\end{theorem}

\begin{proof}
  For $\vec{\alpha}=\mathbf{0}$, our model is a subcase of correlated
  stable matching, so a stable matching $M$ exists. All we need to
  prove now is that the same $M$ is stable when we have friendship
  utilities.

  Suppose it is not the case, i.e., $M$ is unstable for some value of
  $\vec{\alpha}$. This is possible only if we have a blocking pair
  $(u,v)$. But this cannot happen because:
  \begin{itemize}
  \item If both $u$ and $v$ were unmatched in $M$ then $M$ could not
    have been stable for $\vec{\alpha}=\mathbf{0}$.
  \item If exactly one of $u$ and $v$ is unmatched in $M$, say $u$ is
    matched to $w$ and $v$ is unmatched, then for $(u,v)$ to be a
    blocking pair, $r_{uv}>r_{uw}$ by
    Lemma~\ref{lem:increase-edge-reward}. But in such a case, $M$
    could not have been stable for $\vec{\alpha} =\mathbf{0}$.
  \item Suppose both $u$ and $v$ are matched in $M$, say $u$ is
    matched to $w$ and $v$ is matched to $z$. In such a case if
    $(u,v)$ forms a blocking pair corresponding to a biswivel, then by
    Lemma~\ref{lem:increase-edge-reward}, we have $r_{uv} > r_{uw}$
    and $r_{uv} > r_{vz}$ and thus $M$ could not have been stable for
    $\vec{\alpha}= \mathbf{0}$.
  \end{itemize}
  Hence we have shown that no blocking pair exists in $M$ with
  friendship utilities, 
  thus proving the theorem.
\qed \end{proof}


\begin{theorem}\label{thm:poa}
  The price of anarchy in stable matching games with friendship
  utilities is at most 2, and this bound is tight.
\end{theorem}

\begin{proof}
  This theorem is simply a special case of our much more general
  Theorem~\ref{thm:poa-friendship-RS}, which proves a price of anarchy
  bound of $1+\frac{R+\alpha_1}{1+\alpha_1 R}$, with $R$ being a
  measure of how unequally players can share rewards on an edge. When
  players share edge rewards equally, the price of anarchy bound in
  Theorem~\ref{thm:poa-friendship-RS} reduces to
  $1+\frac{1+\alpha_1}{1+\alpha_1}=2$, as desired. To show that this
  bound is tight, simply consider a 3-edge path with all edge rewards
  being 1, for any value of $\vec{\alpha}$.
\qed \end{proof}

\subsection{Price of Stability and Convergence}
\label{subsec:pos}

The main result in this section bounds the price of stability in
stable matching games with friendship utilities to $\frac{2+2\alpha_1}
{1+2\alpha_1+\alpha_2}$, and this bound is tight (see
Theorem~\ref{thm:pos-friendship-equal-sharing} below). This bound has
some interesting implications. It is decreasing in each $\alpha_1$ and
$\alpha_2$, hence having friendship utilities always yields a lower
price of stability than without friendship utilities.  Also, note that
values of $\alpha_3,\alpha_4,...,\alpha_{diam(G)}$ have no
influence. Thus, caring about players more than distance $2$ away does
not improve the price of stability in any way. Also, if
$\alpha_1=\alpha_2=1$, then $\text{PoS}=1$, i.e., there will exist a
stable matching which will also be a social optimum. Thus
\textit{loving thy neighbor and thy neighbor's neighbor but nobody
  beyond} is sufficient to guarantee that there exists at least one
socially optimal stable matching. In fact, due to the shape of the
curve, even small values of friendship quickly decrease the price of
stability; e.g., setting $\alpha_1=\alpha_2=0.1$ already decreases the
price of stability from $2$ to $\sim 1.7$.


We will establish the price of stability bound by defining an
algorithm that creates a good stable matching in polynomial time.
%
%
One possible idea to create a stable matching that is close to optimum
is to use a \BBP\ algorithm: start with the best possible matching,
i.e. a social optimum, which may or may not be stable. Now choose the
``best" blocking pair $(u,v)$: the one with maximum edge reward
$r_{uv}$. Allow this blocking pair to get matched to each other
instead of their current partners. Check if the resulting matching is
stable. If it is not stable then allow the best blocking pair for this
matching to get matched. Repeat the procedure until there are no more
blocking pairs, thereby obtaining a stable matching.

This algorithm gives the desired price of stability and running time
bounds for the case of ``altruism" when all $\alpha_i$ are the same,
see Corollary~\ref{cor:bbp-algorithm} below. To provide the desired
bound with general friendship utilities, we must alter this algorithm
slightly using the concept of \textit{relaxed} blocking pair.

\begin{definition}
  Given a matching $M$, we call a pair of nodes $(u,v)$ a
  \textit{relaxed blocking pair} if either $(u,v)$ form an improving
  swivel, or $u$ and $v$ are matched to $w$ and $z$ respectively, with
  the following inequalities being true:
  \begin{eqnarray}
    (1+\alpha_1)r_{uv} &>& (1+\alpha_1)r_{uw}
                         + (\alpha_1+\alpha_2)\,r_{vz}
	\label{eqn:relaxed-biswivel-1}\\
    (1+\alpha_1)r_{uv} &>& (1+\alpha_1)r_{vz}
   		         + (\alpha_1+\alpha_2)\,r_{uw}
	\label{eqn:relaxed-biswivel-2}
  \end{eqnarray}
\end{definition}

In other words, a relaxed blocking pair ignores the possible edges
between nodes $u$ and $z$, and has $\alpha_2$ in the place of
$\alpha_{uz}$ (similarly, $\alpha_2$ in the place of $\alpha_{vw}$).
It is clear from this definition that a blocking pair is also a
relaxed blocking pair, since the conditions above are less
constraining than Inequalities~\eqref{eqn:biswivel-1}
and~\eqref{eqn:biswivel-2}. Thus a matching with no relaxed blocking
pairs is also a stable matching. Moreover, it is easy to see that
Lemma~\ref{lem:increase-edge-reward} still holds for relaxed blocking
pairs. We will call a relaxed blocking pair satisfying Inequalities
\eqref{eqn:relaxed-biswivel-1} and \eqref{eqn:relaxed-biswivel-2} a
{\em relaxed biswivel,} which may or may not correspond to an
improving deviation, since a relaxed blocking pair is not necessarily
a blocking pair.

Now we present the algorithm to compute a stable matching that is
close to optimal.

\subsubsection{\BRBP\ Algorithm}
\label{subsubsec:bbp-equal-sharing}

\noindent
\begin{enumerate}
  \item	\label{item:initialize}
	Initialize $M=M^*$ where $M^*$ is a socially optimum matching.
  \item \label{item:best-pair-deviate}
	If there is no relaxed blocking pair, terminate. Otherwise make
	the relaxed blocking pair $(u,v)$ with maximum edge reward
	$r_{uv}$ be matched to each other. In other words, remove the edges of $M$ containing $u$ and $v$, and add the edge $(u,v)$ to $M$.
  \item Repeat step~\ref{item:best-pair-deviate}.
\end{enumerate}

\subsubsection{Dynamics of \BRBP}
\label{subsubsec:dynamics-bbp-equal-sharing}


To establish the efficient running time of \BRBP~ and the price of
stability bound of the resulting stable matching, we first analyze the
dynamics of this algorithm and prove some helpful lemmas. We can
interpret the algorithm as a sequence of swivel and relaxed biswivel
deviations, each inserting one edge into $M$, and removing up to two
edges. Note that it is not guaranteed that the inserted edge will stay
forever in $M$, as a subsequent deviation can remove this edge from
$M$. Let $O_1,O_2,O_3,\cdots$ denote this sequence of deviations, and
$e(i)$ denote the edge which got inserted into $M$ because of
$O_i$. Now let us analyze the dynamics of the algorithm by using the
following two lemmas.

\begin{lemma} \label{lem:first-biswivel}
  The first deviation $O_1$ during the execution of \BRBP\ is a
  relaxed biswivel.
\end{lemma}

\begin{proof}
  Having $O_1$ as a swivel will strictly improve the value of matching
  by Lemma~\ref{lem:increase-edge-reward}. Hence if we begin the
  algorithm with $M=M^*$, having $O_1$ as a swivel will produce a
  matching with value strictly greater than $M^*$, which is a
  contradiction.
\qed \end{proof}

\begin{lemma} \label{lem:subsequence}
  Let $O_j$ be a relaxed biswivel that takes place during the
  execution of the best relaxed blocking pair algorithm. Suppose a
  deviation $O_k$ takes place before $O_j$. Then we have $r_{e(k)}\geq
  r_{e(j)}$.  Furthermore, if $O_k$ is a relaxed biswivel then
  $e(k)\neq e(j)$ (thus at most $|E(G)|$ relaxed biswivels can take
  place during the execution of the algorithm).
\end{lemma}

It is important to note that this lemma does \textit{not} say that
$r_{e(i)}\geq r_{e(j)}$ for $i < j$. We are only guaranteed that
$r_{e(i)}\geq r_{e(j)}$ for $i < j$ if $O_j$ is a {\em relaxed
  biswivel}.  Between two successive relaxed biswivels $O_k$ and
$O_j$, the sequence of $r_{e(i)}$ for consecutive swivels can and does
increase as well as decrease, and the same edge may be added to the
matching multiple times. All that is guaranteed is that $r_{e(j)}$ for
a biswivel $O_j$ will have a lower value than all the preceding
$r_{e(i)}$'s. Thus, this lemma suggests a nice representation of
\BRBP\ in terms of phases, where we define a \textit{phase} as a
subsequence of deviations that begins with a relaxed biswivel and
continues until the next relaxed biswivel. Lemma~\ref{lem:subsequence}
guarantees that at the start of each phase, the $r_{e(j)}$ value is
smaller than the values in all previous phases, and that there is only
a polynomial number of phases. Now we proceed to prove
Lemma~\ref{lem:subsequence}.

\begin{proof}
  Let $e(j)=(vz)$ get inserted in $M$ because of a relaxed biswivel
  $O_j$.  We first give a brief outline of the proof. Suppose that the
  claim $r_{e(k)} \geq r_{e(j)}$ for $k<j$ is false and we have an
  $O_k$ with $k<j$ such that $r_{e(k)} < r_{e(j)}$. Clearly $(v,z)$
  could not have been a relaxed blocking pair just before $O_k$,
  otherwise the algorithm would have chosen $(v,z)$ as the best
  relaxed blocking pair instead of $O_k$. We will show that this leads
  to a conclusion that $(v,z)$ cannot be a relaxed blocking pair even
  for $O_j$. This is a contradiction, hence our assumption of
  $r_{e(k)}<r_{e(j)}$ could not have been correct. Thus for all $O_k$
  such that $k<j$ we will have $r_{e(k)}\geq r_{e(j)}$. Later we will
  use similar reasoning to prove that if $O_i$ with $i<j$ is a relaxed
  biswivel that takes place before a relaxed biswivel $O_j$ then
  $e(i)\neq e(j)$. Now let us proceed to the proof.

  Suppose to the contrary that we have $O_k$ with $k<j$ such that
  $r_{e(k)} < r_{e(j)}$ with $O_j$ being a relaxed biswivel. As
  discussed in the outline of the proof, this implies that $(v,z)$ was
  not a relaxed blocking pair at the time $O_k$ was selected. Let $S$
  be the set of nodes with whom $v$ and $z$ are matched at the time
  that $O_k$ is selected.  As long as $S$ does not change, $v$ and $z$
  will not be a relaxed blocking pair, since the change in utility
  experienced by $v$ and $z$ from matching to each other depends only
  on their partners in the current matching, i.e., the set $S$. Thus
  for the relaxed biswivel $O_j$ to occur, $S$ must change between
  $O_k$ and $O_j$. We will show that this leads to a contradiction:
  that $(v,z)$ cannot be a relaxed blocking pair for the time $O_j$ is
  selected.

  Suppose $v$ is matched to $x$ and $z$ is matched to $y$ just before
  biswivel $O_j$. Since $(v,z)$ is a relaxed blocking pair at this
  point, we thus have
  \begin{eqnarray}
    (1+\alpha_1)r_{vz} &>& (1+\alpha_1) r_{vx} +
    (\alpha_1+\alpha_2) r_{zy} \label{eqn:support-5}\\
    (1+\alpha_1)r_{vz} &>& (1+\alpha_1) r_{zy} +
    (\alpha_1+\alpha_2) r_{vx}. \label{eqn:support-6}
  \end{eqnarray}
  Recall that $(v,z)$ was not a relaxed blocking pair just before
  $O_k$, and to make it a relaxed blocking pair for $O_j$, $S$ must
  change between $O_k$ and $O_j$. Let $O_l$ be the last deviation
  which changed $S$ to $\{x,z\}$. Without loss of generality, we can
  assume that $O_l$ adds the edge $(v,x)$.  Now we have two cases:
  \begin{itemize}
  \item	$(v,z)$ was a relaxed blocking pair at the time $O_{l}$ is selected:
	in this case $(v,x)$ could not have been the best relaxed
	blocking pair for $O_l$ because
	inequality~\eqref{eqn:support-5} tells us $r_{vz} > r_{vx}$.
  \item	$(v,z)$ was not a relaxed blocking pair at the time $O_{l}$
	is selected: Suppose $v$ was matched with $w$ before $O_l$.
	As $(v,z)$ was not a relaxed blocking pair just before $O_l$
	we have
	\begin{eqnarray}
	\text{Either   }
  	(1+\alpha_1)r_{vz} &\leq& (1+\alpha_1) r_{vw} +
		(\alpha_1+\alpha_2) r_{zy} \label{eqn:support-501}\\
	\text{OR   }
  	(1+\alpha_1)r_{vz} &\leq& (1+\alpha_1) r_{zy} +
		(\alpha_1+\alpha_2) r_{vw} \label{eqn:support-502}
	\end{eqnarray}
	(If $v$ was not matched just before $O_l$ then substitute
	$r_{vw}=0$ to obtain appropriate condition.) Assume that
	it is inequality~\eqref{eqn:support-501} that holds.
	Then, because $O_l$ removes edge $(v,w)$ and adds edge $(v,x)$,
	we have $r_{vx} > r_{vw}$ as
	Lemma~\ref{lem:increase-edge-reward} holds for relaxed blocking
	pairs. Thus, the following must be true:
	\begin{eqnarray}
  	(1+\alpha_1)r_{vz} &\leq& (1+\alpha_1) r_{vx} +
		(\alpha_1+\alpha_2) r_{zy} \label{eqn:support-503}
	\end{eqnarray}
    	This contradicts inequality~\eqref{eqn:support-5}, and thus
   	$(v,z)$ cannot be a relaxed blocking pair at the time $O_j$ is
	selected. The same conclusion can be reached if we assume
	inequality~\eqref{eqn:support-502} holds true.
  \end{itemize}
  Either way we arrive at a contradiction, thus showing that if $O_j$
  is a relaxed biswivel then for all $O_k$ with $k<j$, we have
  $r_{e(k)} < r_{e(j)}$.

  Now the only remaining piece is to prove $e(k)\neq e(j)$ if $O_k$ is
  a relaxed biswivel.  All we need to notice that if $e(k)=e(j)=(v,z)$
  then $S$ has to change between $O_k$ and $O_j$. Now we use exactly
  the reasoning from the previous paragraph to arrive at a
  contradiction, thus proving that $e(k)\neq e(j)$.
\qed \end{proof}

\subsubsection{Convergence of \BRBP}
\label{subsubsec:convergence-bbp-equal-sharing}

For the case where $\alpha_1=\alpha_2$, the conditions for a blocking
pair are identical to the conditions for a relaxed blocking
pair. Hence, our algorithm corresponds to letting the best blocking
pair deviate at each step. As a special case, for
$\vec{\alpha}=\mathbf{0}$ and correlated stable matching, this
algorithm is known to provide a stable matching in polynomial
time~\cite{Ackermann11}. For friendship utilities, however, (quick)
convergence was previously unknown. Here we will show that even with
the addition of friendship, \BRBP\ (and thus \BBP\ for the case when
$\alpha_1=\alpha_2$) terminates and produces a stable
matching. Moreover, it does this in polynomial time.


Note that if instead of the best we pick some arbitrary blocking pair,
then there exists an instance in which, starting from the empty
matching, a sequence of blocking pairs of length $2^{\Omega(n)}$
exists until reaching a stable matching, even without friendship. This
is directly implied by recent results in correlated stable
matching~\cite{HoeferICALP11}.

A trivial adjustment of the gadget in~\cite{HoeferICALP11} allows us
to construct the exponential sequence even when starting from the
social optimum. We scale the reward of each (original) edge $i \in
\{1,\ldots,m\}$ in the gadget from $i$ to $1+i\cdot\epsilon$, for some
tiny $\epsilon > 0$. This preserves all incentives and the structure
of all blocking pairs. Then, we add an auxiliary neighbor for each
(original) player and connect it via an auxiliary edge of reward
1. The social optimum is obviously given by matching each original
player with his auxiliary neighbor. However, the exponential sequence
of blocking pairs still exists, as auxiliary edges are not rewarding
enough to influence blocking pairs among original players. Due to the
fact that such exponential-length sequences exist, it is perhaps
surprising that our algorithm indeed finds a stable matching and it
terminates in polynomial time.

\begin{theorem} \label{thm:convergence-bbp}
  \BRBP\ outputs a stable matching after $O(m^2)$ iterations, where
  $m$ is the number of edges in the graph.
\end{theorem}

\begin{proof}
  Consider the three possible changes that can occur to the matching
  $M$ during each iteration: a swivel could add a new edge, or it
  could delete an edge and add an edge with strictly higher $r_e$
  value. A relaxed biswivel deletes two edges, and adds an edge with
  higher $r_e$ value than either. Thus, without any biswivels taking
  place, the total number of consecutive swivels is at most $m^2$,
  since no edges are deleted by swivels. Each relaxed biswivel can
  allow at most $m$ extra swivels to occur, since it deletes one
  edge. As there are at most $m$ relaxed biswivel deviations possible
  by Lemma~\ref{lem:subsequence}, the algorithm terminates after at
  most $m^2+m^2$ deviations. Since there are no more relaxed blocking
  pairs for the algorithm to continue, and since a blocking pair is
  also a relaxed blocking pair, then the final matching produced by
  the algorithm is a stable matching.
\qed \end{proof}


As we can have only a polynomial number of consecutive swivel
deviations between each relaxed biswivel, we know that every phase
(defined as a maximal subsequence of consecutive swivels) lasts only a
polynomial amount of time, and there are only $O(m)$ phases by
Lemma~\ref{lem:subsequence}. Moreover, in each phase, the value of the
matching only increases, since swivels only remove an edge if they add
a better one. Below, we use the fact that only relaxed biswivel
operations reduce the cost of the matching to bound the cost of the
stable matching this algorithm produces.

\subsubsection{Upper Bound on Price of Stability}
\label{subsubsec:upper-bound-pos}

Before proceeding to prove the bound, we will introduce some notation
and prove some useful lemmas.

We define a sequence of mappings from $M^*$ to $E(G)$. Define
$h_0:M^*\rightarrow E(G)$ as $h_0(e)=e$. Depending on $O_i$, we will
define $h_i$ as follows:
Suppose $O_i$ is a deviation that removes edge $h_{i-1}(e_j)$ from
$M$. If $O_i$ inserts edge $e_l$ in $M$ then set $h_i(e_j) = e_l$. For
all other $e_k\in M^*$, keep $h_{i}(e_k)$ same as $h_{i-1}(e_k)$. Let
us note that a deviation $O_i$ may not remove any edges from
$\{h_{i-1}(e_j): e_j\in M^*\}$. This can happen because during the
course of the algorithm, two unmatched nodes can get matched, say to
insert $e_p$ into $M$. No edges in $M^*$ get mapped to $e_p$. If this
edge is removed from $M$ by a later deviation, the mapping may not
change, since no edge is mapped to $e_p$. To summarize, $h_i$ may be
the same as $h_{i-1}$, or may differ from $h_{i-1}$ in one location
(in case of a swivel), or in two locations (in case of a relaxed
biswivel).
Denote the resulting mapping when our algorithm terminates by $h_M$.

Coupling Lemma~\ref{lem:increase-edge-reward} with the definition of
mappings $h_i$, we immediately have the following result:
\begin{lemma} \label{lem:increasing-sequence}
  $\{r_{h_i(e)}\}_{i\geq 0}$ is a nondecreasing sequence and
  $r_{h_{i+1}(e)} > r_{h_i(e)}$ whenever $h_{i+1}(e)\neq h_i(e)$.
\end{lemma}

The following lemma will be instrumental in proving the price of
stability bound.

\begin{lemma} \label{lem:final-mapping}
  If $h_M(e_i)=h_M(e_j)$ with $e_i\neq e_j$ then
  \begin{enumerate}
  \item	There must exist a relaxed biswivel $O_k$ such that
	$h_{k-1}(e_i)\neq h_{k-1}(e_j)$ but $O_k$
	makes $h_k(e_i)=h_k(e_j)$. Furthermore,
	for all $p\geq k$ we have $h_p(e_i)=h_p(e_j)$.
  \item	There does not exist another $e_l\in M^*$ such that
	$h_M(e_l)=h_M(e_i)=h_M(e_j)$.
  \item \label{item:decrease-biswivel}
	$r_{e_i} + r_{e_j} <
	\frac{2+2\alpha_1}{1+2\alpha_1+\alpha_2}\times r_{h_M(e_i)}$
  \end{enumerate}
\end{lemma}

\begin{proof}
  To prove the first part, say $O_l$ was the first deviation such that
  $h_{l-1}(e_i)\neq h_{l-1}(e_j)$ and $h_l(e_i)=h_l(e_j)$. It cannot
  happen because of a swivel deviation because a swivel can make
  $h_l(e)\neq h_{l-1}(e)$ for at most for one $e\in M^*$. Thus $O_l$
  must be a relaxed biswivel. Set $k=l$ and it is easy to see that for
  $p \geq k$ we have $h_p(e_i)=h_p(e_j)$.  Hence the first part is
  proven.

  To prove the second part, suppose there exists an $e_l$ with
  $e_l\neq e_i\neq e_j$ such that $h_M(e_l)=h_M(e_i)=h_M(e_j)$. From
  the first part, there must exist a relaxed biswivel $O_k$
  s.t. $h_{k-1}(e_i)\neq h_{k-1}(e_l)$ but $h_k(e_i)=h_k(e_l)$.
  Similarly there must exist a relaxed biswivel $O_p$ s.t.
  $h_{p-1}(e_i)\neq h_{p-1}(e_j)$ but $h_p(e_i)=h_p(e_j)$.  Without
  loss of generality say $p>k$. Using Lemma~\ref{lem:subsequence} we
  get $r_{e(k)}\geq r_{e(p)}$.  But from
  Lemma~\ref{lem:increasing-sequence}, we have $r_{e(k)}<r_{e(p)}$,
  since $e(p)=h_p(e_i) > h_k(e_i)=e(k)$.  Hence we have a
  contradiction here, thus proving that there does not exist another
  $e_l\in M^*$, with $h_M(e_l)=h_M(e_i)=h_M(e_j)$.

  To prove the third part, consider a relaxed biswivel $O_k$ such that
  $h_{k-1}(e_i)\neq h_{k-1}(e_j)$ and $h_k(e_i)=h_k(e_j)$. Substitute
  $r_{uv} = r_{h_k(e_i)}$, $r_{uw} = r_{h_{k-1}(e_i)}$ and
  $r_{vz}=r_{h_{k-1}(e_j)}$ in inequalities~\eqref{eqn:biswivel-1}
  and~\eqref{eqn:biswivel-2}. Adding these inequalities and
  simplifying, we get
  \begin{eqnarray} \label{eqn:bound-intermediate} r_{h_{k-1}(e_i)} +
    r_{h_{k-1}(e_j)} < \frac{2+2\alpha_1}{1+2\alpha_1+\alpha_2} \times
    r_{h_k(e_i)}
  \end{eqnarray}
  From Lemma~\ref{lem:increasing-sequence}, we have that
  $\{r_{h_i(e)}\}_{i\geq 0}$ as a nondecreasing sequence. Using this
  in~\eqref{eqn:bound-intermediate} we get
  \begin{eqnarray} \label{eqn:decrease-biswivel}
    r_{e_i} + r_{e_j}<
    \frac{2+2\alpha_1}{1+2\alpha_1+\alpha_2}\times r_{h_M(e_i)}
  \end{eqnarray}
\qed \end{proof}

Using Lemma~\ref{lem:final-mapping}, we can partition edges of $M^*$
into two sets as follows: Let $B$ denote the set of edges $e_i\in M^*$
such that $h_M(e_i)=h_M(e_j)$ for some $e_j\in M^*$ and let $A$ denote
the remaining edges in $M^*$. We can further partition set $B$ into
two sets $P$ and $Q$ as follows: choose a pair $e_i$ and $e_j$ in $B$
such that $h_M(e_i)=h_M(e_j)$. Denote $e_j$ by $\mu(e_i)$. Put $e_i$
in $P$ and $\mu(e_i)$ in $Q$.  Notice that value of the matching $M$
that \BRBP\ gives as output is at least $\sum_{e\in A}r_{h_M(e)} +
\sum_{e\in P}r_{h_M(e)}$. (The possible additional edges in $M$ are
produced because of swivels which match two previously unmatched nodes
with each other.)

We are now in position to prove the main theorem of this section:

\begin{theorem}\label{thm:pos-friendship-equal-sharing}
  The price of stability in stable matching games with friendship
  utilities is at most $\frac{2+2\alpha_1}{1+2\alpha_1+\alpha_2}$, and
  this bound is tight.
\end{theorem}

\begin{proof}
  The value of $M^*$ is given by
  \begin{eqnarray}
    w(M^*) &=& \sum\limits_{e\in A} r_e
	     + \sum\limits_{e\in P} r_e
	     + \sum\limits_{e\in Q} r_e \nonumber \\
           &=& \sum\limits_{e\in A} r_e
	     + \sum\limits_{e\in P} (r_e + r_{\mu(e)}) \nonumber
  \end{eqnarray}
  Using Lemma~\ref{lem:final-mapping}, for $e\in P$ we have
  $r_e+r_{\mu(e)}<\frac{2+2\alpha_1}{1+2\alpha_1+\alpha_2} \times
  r_{h_M(e)}$.  Using Lemma~\ref{lem:increasing-sequence}, for $e\in
  A$ we have $r_e \leq r_{h_M(e)}$. Thus we get
  \begin{eqnarray*}
 w(M^*) &\leq&  \sum\limits_{e\in A} r_{h_M(e)}
	 + 	\sum\limits_{e\in P}
		\frac{2+2\alpha_1}{1+2\alpha_1+\alpha_2}
	\times  r_{h_M(e)} \nonumber \\
	&\leq&  \frac{2+2\alpha_1}{1+2\alpha_1+\alpha_2}
		\left(
		\sum\limits_{e\in A} r_{h_M(e)}
	 + 	\sum\limits_{e\in P} r_{h_M(e)}
		\right)
  \end{eqnarray*}
  Note that this inequality may not be strict since $A$ may be
  empty. This could happen if each edge in $M^*$ gets removed because
  of a relaxed biswivel as the algorithm proceeds (though it may be
  possible that it is inserted later). We also have $w(M) \geq
  \sum_{e\in A} r_{h_M(e)} + \sum_{e\in P} r_{h_M(e)}$ for the final
  matching $M$ that the algorithm gives. Using this,
  \begin{eqnarray*}
    w(M^*) \leq \frac{2+2\alpha_1}{1+2\alpha_1+\alpha_2} w(M),
  \end{eqnarray*}
  which proves the bound on the price of stability, since $M$ is a
  stable matching.

  To prove the tightness of the bound, let us make $\alpha_2=0$ and
  assign $r_{uv} = \frac{1+2\alpha_1+\epsilon}{1+\alpha_1}$, $r_{uw} =
  r_{vz} = 1$ in Fig~\ref{fig:biswivel}. Then we have $\{(uv)\}$ as
  the only stable matching but the social optimum is
  $\{(uw),(vz)\}$. Thus we get
  PoS$=\frac{2+2\alpha_1}{1+2\alpha_1+\epsilon}$ which can be taken
  arbitrarily close to $\frac{2+2\alpha_1}{1+2\alpha_1}$.  This gives
  us a tight bound given that we are using $\alpha_2=0$.
\qed \end{proof}

From Theorems~\ref{thm:convergence-bbp}
and~\ref{thm:pos-friendship-equal-sharing}, we immediately get the
following corollary about the behavior of best blocking pair
dynamics. This corollary applies in particular to the model of
altruism when $\alpha_i = \alpha$ for all $i=1,\ldots,diam(G)$.
\begin{corollary}\label{cor:bbp-algorithm}
  If $\alpha_1=\alpha_2$ and we start from the centrally optimum
  matching, \BBP\ converges in $O(m^2)$ time to a stable matching that
  is at most a factor of $\frac{2+2\alpha_1}{1+2\alpha_1+\alpha_2}$
  worse than the optimum.
\end{corollary}
\begin{proof}
  Immediate since when $\alpha_1=\alpha_2$, the conditions for a
  blocking pair are identical to the conditions for a relaxed blocking
  pair. Hence, \BRBP\ is \BBP.
\qed \end{proof}


%% file: reward-sharing.tex
\section{General Reward Sharing without Friendship}
\label{sec:reward-sharing}

In the previous section we considered the case where if $(uv)\in M$ then $u$ and $v$
get the \textit{same reward} from edge $(uv)$, namely $r_{uv}$.
Now we will look into
the case where $u$ and $v$ may possibly share the edge reward
$r_{uv}$ if $(uv)\in M$. Let us define $r^x_{xy}$ as the reward
node $x$ gets from edge $(xy)$ if $(xy)\in M$. We assume $r_{xy} = r^x_{xy} + r^y_{xy}$ as $x$ and $y$ share the edge
reward $r_{xy}$. Note that while in Section~\ref{sec:friendship-equal-sharing}
we had both nodes $u$ and $v$ getting reward of $r_{uv}$ from edge
$(uv)$, this is actually equivalent to $u$ and $v$ sharing the edge reward equally. To see this  we can redefine the reward $u$ and $v$ get from edge $(uv)$ as $\frac{r_{uv}}{2}$ if $(uv)\in M$. This scaling of edge reward
does not affect any of the results.

For general (unequal) reward sharing, we will give results about
existence of a stable matching, as well as bounds on prices of anarchy
and stability. In addition to that, we will also focus on the
following specific reward sharing rules:
\begin{itemize}
  \item \label{item:matthew-effect} \textit{Matthew Effect sharing:} In
        sociology, ``Matthew Effect'' is a term coined by Robert
        Merton 
        to describe the phenomenon which says that, when doing similar work, the more famous person tends to get more credit than other less-known collaborators. We model such phenomena for our network by associating brand values $\lambda_u$ with each node $u$, and defining the reward that node $u$ gets by
        getting matched with node $v$ as $r^u_{uv} =
        \frac{\lambda_u}{\lambda_u+\lambda_v} \cdot r_{uv}$. Thus nodes
        $u$ and $v$ split the edge reward in the ratio of
        $\lambda_u:\lambda_v$, and a node with high $\lambda_u$ value gets a disproportionate amount of reward.
  \item	\label{item:parasite}
	\textit{Parasite sharing:}
	This effect is opposite to the Matthew effect in the sense
	that by collaborating with a renowned person, a less-known
	person becomes famous, whereas the reputation of the already
	renowned person does not change significantly from such a
	collaboration. We model this situation by defining
	the reward that node $u$ gets by getting matched with node
	$v$ as $r^u_{uv} = \frac{\lambda_v}{\lambda_u+\lambda_v}
	r_{uv}$. Thus nodes $u$ and $v$ split the edge reward
	in the ratio of $\lambda_v:\lambda_u$, in the exactly
	opposite way to the Matthew Effect sharing.
  \item	\label{item:trust}
	\textit{Trust sharing:}
	Often people collaborate based on not only the quality
	of a project but also how much they trust each
	other. We model such a situation by associating a value $\beta_u$ with each node $u$,
    which represents the \textit{trust value} of player $u$, or how pleasant they are to work with.
    Each edge $(u,v)$ also has an inherent quality $h_{uv}$. Then, the reward obtained by node $u$
    from partnering with node $v$ is
	$r^u_{uv} = h_{uv} + \beta_v$.
\end{itemize}

For the sake of analysis, Matthew Effect sharing and Parasite sharing
are the same if we change $\lambda_u$ of Parasite sharing to
$1/\lambda_u$ of Matthew Effect sharing. We will refer to both the
models as Matthew Effect sharing from now on. In the next few
sections, we will give results about the existence of stable
matchings, and give upper bounds on prices of anarchy and stability
for Matthew Effect sharing and Trust sharing, as well as for general
reward sharing. Note that this analysis is for the case when
friendship is absent; we consider the more general case of unequal
sharing with friendship utilities in Section~\ref{sec:friendship-RS}.

\subsection{Existence of a Stable Matching}
\label{subsec:stable-matching-RS}

Without friendship utilities, our stable matching game reduces to the
stable roommate problem, since reward shares can be arbitrary and thus
induce arbitrary preference lists for each node. It is well known that
a stable matching may not exist in a stable roommate
problem~\cite{Gale62}. However, we will prove in this section that for
Matthew Effect sharing and Trust sharing, a stable matching can always
be found.


Let us define a \textit{preference cycle} as a cycle $(u_1, u_2,
\cdots, u_k)$ in the graph $G$ such that $r^{u_i}_{u_i u_{i+1}} \geq
r^{u_i}_{u_i u_{i-1}}$ with at least one inequality being strict.
Chung~\cite{Chung00} defines \textit{odd rings} and proves that if a
graph does not contain odd rings, then a stable matching exists.  It
is straightforward to see that absence of preference cycles implies
absence of odd rings. Hence, if a graph has no preference cycles, then
a stable matching must exist. Below we prove the stronger statement
that such a matching can also be found efficiently.

\begin{theorem}\label{thm:ring-cycle}
  A stable matching always exists in stable matching games with
  unequal sharing and no preference cycles. Furthermore, a stable
  matching can be found in $O(|V||E|)$ time.
\end{theorem}

\begin{proof}
  In brief, we show below that whenever there exist no preference
  cycles in a graph, we can always find two nodes which prefer getting
  matched to each other over other nodes. We allow them to get matched
  to each other and eliminate such matched nodes from the
  graph. Neither of these two nodes will ever deviate from this
  matching. Applying the same greedy scheme on the reduced graph will
  give us a stable matching. Then we will prove that this algorithm
  produces a stable matching in $O(|V||E|)$ time. Let us now proceed
  to the details.

  Let $T_u$ denote the sets of ``best" neighbors of $u$ as follows:
  \begin{eqnarray}
    T_u = \{v\in N_1(u): r_{uv}^u \geq r_{uw}^u \,\, \forall (uw)\in G\}
  \end{eqnarray}
  Now we construct a directed graph $G_D$ as follows: for all nodes
  $u$, choose a node $v\in T_u$ and draw an edge from $u$ directed to
  $v$. Every node in this graph has one outgoing edge hence this graph
  contains a (directed) cycle. If we find a cycle of length $2$ then
  we have found two nodes which prefer each other the most. If a
  (directed) cycle $(u_1,u_2,\ldots,u_k)$ has length $k>2$, then we
  have $r^{u_i}_{u_i u_{i+1}} \geq r^{u_i}_{u_i u_{i-1}}$.  Now we
  cannot have $r^{u_2}_{u_2u_3} > r^{u_2}_{u_1u_2}$, otherwise in the
  original graph $G$, $(u_1,u_2,\ldots,u_k)$ would have constituted a
  preference cycle. Hence we have $r^{u_2}_{u_1u_2} = r^{u_2}_{u_2
    u_3}$. Thus $u_1$ and $u_3$ both are $u_2$'s most preferred
  nodes. But we also have $u_1$ prefer $u_2$ the most as $G_D$ has an
  edge from $u_1$ to $u_2$. Hence $u_1$ and $u_2$ is the pair of nodes
  that prefer each other the most.

  Therefore we will always be able to find two nodes in $G$ which
  prefer each other the most in their preference lists. Match them to
  each other and they will never have incentive to deviate from this
  matching. Remove these two nodes and repeat the procedure until no
  more nodes can be matched. Because no nodes matched in this process
  will ever deviate, we have a stable matching.

  It takes $O(|E|)$ time to find each matched pair because for each
  edge we check if two nodes prefer each other the most. Since total
  number of nodes to be matched are $O(|V|)$, we find a stable
  matching in $O(|V||E|)$ time.
\qed \end{proof}

\noindent
Now we can prove the following theorem:

\begin{theorem}\label{thm:cycle-MEM-TM}
  No preference cycles exist with Matthew Effect sharing and Trust
  sharing. Hence, a stable matching exists with Matthew Effect sharing
  and Trust sharing and can be found efficiently.
\end{theorem}

\begin{proof}
  Suppose a preference cycle exists in Matthew Effect sharing. Then
  there exists a cycle $(u_1,u_2,\ldots,u_k)$ such that
  \begin{eqnarray}
    \frac{\lambda_{u_i}}{\lambda_{u_i}+\lambda_{u_{i+1}}}r_{u_i u_{i+1}}
    &\geq&
    \frac{\lambda_{u_i}}{\lambda_{u_i}+\lambda_{u_{i-1}}}r_{u_i u_{i-1}}
  \end{eqnarray}
  with at least one inequality being strict.  Multiplying all these
  inequalities and canceling common factors, we reach a contradiction
  that $1>1$. Thus a preference cycle cannot exist in Matthew Effect
  sharing.

  Suppose a preference exists in Trust sharing. Then there exists a
  cycle $(u_1,u_2,\ldots,u_k)$ such that
  \begin{eqnarray}
    h_{u_i u_{i+1}} + \beta_{u_{i+1}}
    &\geq&
    h_{u_i u_{i-1}} + \beta_{u_{i-1}}
  \end{eqnarray}
  with at least one inequality being strict. Adding all these
  inequalities and canceling common factors, we reach a contradiction
  that $0>0$. Thus a preference cycle cannot exists in Trust sharing.

  Since preference cycles cannot exist, we only need to apply
  Theorem~\ref{thm:ring-cycle} to obtain the desired result.
\qed \end{proof}

\subsection{Prices of Anarchy and Stability with General Reward Sharing}
\label{subsec:poa-pos-reward-sharing-without-friendship}

In this section, we will investigate prices of anarchy and stability
with general reward sharing. First we will prove that for general
reward sharing, the price of anarchy is upper bounded by
$1+\max_{(uv)\in E(G)}\frac{r^u_{uv}}{r^v_{uv}}$. This implies a bound
of $1+\max_{(uv)\in E(G)}\frac{\lambda_u}{\lambda_v}$ for Matthew
Effect sharing. We will further prove that for the special case of
Trust sharing, the upper bound on the price of anarchy is $3$.

Let us define $R$ as
\begin{eqnarray}
  R = \max_{(uv)\in G} \cfrac{r^u_{uv}}{r^v_{uv}}
\end{eqnarray}
Note that we will always have $R\geq 1$. We have the following
theorem:
\begin{theorem} \label{thm:poa-RS}
  If a stable matching exists, both prices of anarchy and stability
  in stable matching games with unequal reward sharing (without
  friendship utilities) are at most $R+1$ and the bound is tight.
\end{theorem}

The tightness of this bound implies that as sharing becomes more
unfair, i.e., as $R\rightarrow \infty$, we can find instances where
both prices are unbounded. Thus unequal sharing can make things much
worse for the stable matching game. In
Section~\ref{sec:friendship-RS}, however, we will see that this bound
will significantly improve if we introduce friendship utilities. Thus,
caring about others when reward sharing is unfair makes a significant
difference to the price of anarchy, much more so that in equal
sharing.

Now let us proceed to the proof of Theorem~\ref{thm:poa-RS}.

\begin{proof}
  This theorem is simply a special case of our much more general
  Theorem~\ref{thm:poa-friendship-RS}, which proves a price of anarchy
  bound of $1+\frac{R+\alpha_1}{1+\alpha_1 R}$. Without friendship
  utilities, the price of anarchy bound in
  Theorem~\ref{thm:poa-friendship-RS} reduces to $1+\frac{R}{1}=1+R$,
  as desired.
  To show that this bound is tight, we will use an instance of Matthew
  Effect
  sharing. 
  We assign the following values in Fig.~\ref{fig:biswivel}:
  $r_{uv}=2$, $r_{uw}=R+1$, $r_{vz}=R+1$. Let $\lambda_u=1$,
  $\lambda_v=1$, $\lambda_w=R$, $\lambda_z=R$. Thus we have
  $r^u_{uv}=r^u_{uw}=r^v_{uv}=r^v_{vz}=1$.  Now the matching
  $\{(uv)\}$ is stable and hence we get $\text{PoA}=R+1$.  For
  tightness of PoS bound, change $r_{uv}$ to $2+2\ep$.  Now we have
  $r^u_{uv}=r^v_{uv}=1+\epsilon$ but $r^u_{uw}=r^v_{vz}=1$, hence the
  matching $\{(uv)\}$ is the only stable matching. Thus we get
  $\text{PoS}=\frac{R+1}{1+\ep}$ which can be taken arbitrarily close
  to $R+1$.
\qed \end{proof}

\begin{theorem} \label{thm:poa-trust}
  The price of anarchy in stable matching games with Trust sharing is
  at most $3$.
\end{theorem}

\begin{proof}
  A stable matching always exists for Trust sharing by
  Theorem~\ref{thm:cycle-MEM-TM}. Now we will prove that the price of
  anarchy can be at most $3$.

  Let $M^*$ denote a socially optimum matching and let $M$ denote a
  stable matching. Let $w^*_u$ denote the reward a node $u$ gets in
  $M^*$ and $w_u$ denote the reward a node $u$ gets in $M$. Consider
  an edge $(uv)\in M^*\setminus M$.  As $(u,v)$ is not a blocking pair
  in $M$, without loss of generality, we can assume that the utility
  of $u$ does not increase by getting matched with $v$ in $M$.  Now
  $u$ must be matched to some other node, say $z$.  Call $u$ a
  \textit{witness} node for $(uv)\in M^*\setminus M$. Since $u$ does
  not want to switch to $(uv)$ from $M$, then
  \begin{eqnarray} \label{eqn:support-101}
    h_{uv} + \beta_v \leq h_{uz} + \beta_z
  \end{eqnarray}
  Adding $\beta_u$ to both sides
  \begin{eqnarray} \label{eqn:support-111}
    h_{uv} + \beta_v + \beta_u \leq h_{uz} + \beta_z + \beta_u
  \end{eqnarray}
  From Inequality~\eqref{eqn:support-101}, we get $h_{uv} \leq h_{uz}
  + \beta_z$. Addding this to Inequality~\eqref{eqn:support-111}, we
  obtain
  \begin{eqnarray}
    2h_{uv} + \beta_v + \beta_u \leq 2h_{uz} + 2\beta_z + \beta_u
  \end{eqnarray}
  Suppose we form such inequalities for all $(uv)\in M^*\setminus M$
  and add them. Let us investigate the coeffcients of terms appearing
  on right hand side after such addition. If a term $h_{uz}$ appears
  on right hand side, then its coeffcient can be at most $4$: counting
  one inequality for $u$ acting as witness, and possibly one more
  inequality for $z$ acting as witness. However, the coeffcient for a
  term $\beta_z$ appearing on right hand side can be at most $3$,
  because $2\beta_z$ comes in when $u$ acts as witness and $\beta_z$
  comes in when $z$ acts as witness. Hence adding these inequalities
  will give us
  \begin{eqnarray}
  	\sum\limits_{(uv)\in M^*\setminus M}
	2h_{uv} + \beta_v + \beta_u
   &\leq& \sum\limits_{(uv)\in M\setminus M^*}
	4h_{uv} + 3\beta_v + 3\beta_u \nonumber \\
   \Rightarrow
  	\sum\limits_{(uv)\in M^*\setminus M}
	2h_{uv} + \beta_v + \beta_u
   &\leq& 3\sum\limits_{(uv)\in M\setminus M^*}
	2h_{uv} + \beta_v + \beta_u \label{eqn:inequality-120}
  \end{eqnarray}
  But $r_{xy} = r^x_{xy}+r^y_{xy} = 2h_{xy} + \beta_x + \beta_y$ for
  Trust sharing. Substituting this in
  inequality~\eqref{eqn:inequality-120}, we get
  \begin{eqnarray}
  	\sum\limits_{(uv)\in M^*\setminus M} r_{uv}
  &\leq&  \,\,\, 3 \sum\limits_{(uv)\in M\setminus M^*} r_{uv}
	\nonumber \\
  \Rightarrow
  \frac{w(M^*)}{w(M)} &=&
	\frac{\sum\limits_{(uw)\in  M^*} r_{uw}}
	     {\sum\limits_{(uv)\in  M} r_{uv}}
	\leq 3 \nonumber
  \end{eqnarray}
  As this is valid for any stable matching $M$, we have proved that
  for stable matching games with trust sharing, $\text{PoA} \leq 3$.
\qed \end{proof}

%% file: fship-reward-sharing.tex
\section{Stable Matching with Friendship and General Reward Sharing}
\label{sec:friendship-RS}

In this section, we consider general stable matching games where
players may have both friendship utilities and unequal reward
sharing. We show general bounds on price of anarchy, and establish
that friendship can make a much larger difference in the context of
unequal sharing than in the case of fair sharing. First, just as at
the start of Section~\ref{sec:friendship-equal-sharing}, we write
explicit conditions for nodes to form a blocking pair in this context,
and define some helpful notation.

The necessary and sufficient conditions for nodes $(u,v)$ to form a
biswivel from nodes $w$ and $z$ (See Fig.~\ref{fig:biswivel}) in reward
sharing with friendship are:
\begin{eqnarray}
  r^u_{uv} + \alpha_1 r^v_{uv} &>&
     r^u_{uw} + \alpha_1 (r^w_{uw} + r^v_{vz}) + \alpha_{uz} r^z_{vz}
		\nonumber\\
  r^v_{uv} + \alpha_1 r^u_{uv} &>&
     r^v_{vz} + \alpha_1 (r^z_{vz} + r^u_{uw}) + \alpha_{vw} r^w_{uw}.
		\nonumber
\end{eqnarray}
Let us define $q^x_{xy} = r^x_{xy} + \alpha_1 r^y_{xy}$. Then the
conditions for biswivel such as shown in Fig.~\ref{fig:biswivel} are:
\begin{eqnarray}
  q^u_{uv} &>& q^u_{uw} + \alpha_1 r^v_{vz} + \alpha_{uz} r^z_{vz}
			  \label{eqn:biswivel-RS-1}\\
  q^v_{uv} &>& q^v_{vz} + \alpha_1 r^u_{uw} + \alpha_{vw} r^w_{uw}.
			  \label{eqn:biswivel-RS-2}
\end{eqnarray}
Similarly, the necessary and sufficient conditions for swivel (See
Fig.~\ref{fig:swivel}) are
\begin{eqnarray}
  r^u_{uv} + \alpha_1 r^v_{uv} &>&
		r^u_{uw} + \alpha_1 r^w_{uw} \nonumber\\
  r^v_{uv} + \alpha_1 r^u_{uv} &>&
		\alpha_1 r^u_{uw} + \alpha_{vw} r^w_{uw}. \nonumber
\end{eqnarray}
Using the definition of $q^x_{xy}(\cdot,\cdot)$, the conditions for
swivel become:
\begin{eqnarray}
 q^u_{uv} &>& q^u_{uw} \label{eqn:swivel-RS} \\
 q^v_{uv} &>& \alpha_1 r^u_{uw} + \alpha_{vw} r^w_{uw}
	\label{eqn:swivel-RS-2}
\end{eqnarray}
We also have
\begin{eqnarray}
  \frac{q^x_{xy}}{q^y_{xy}} =
	\cfrac{r^x_{xy}+\alpha_1 r^y_{xy}}{r^y_{xy}+\alpha_1 r^x_{xy}}
		\nonumber
\end{eqnarray}
Using the fact that $\frac{p+\alpha_1}{1+\alpha_1 p}$ is an increasing
function of $p$ and using the definition of $R$, we thus obtain
\begin{eqnarray}
  \frac{q^x_{xy}}{q^y_{xy}} \leq \frac{R+\alpha_1}{1+\alpha_1 R}
\end{eqnarray}
Let us define $q_{xy} = q^x_{xy} + q^y_{xy}$. Thus we obtain
$q_{xy} = (1+\alpha_1) r_{xy}$.
\subsection{Existence of a Stable Matching with Friendship and General
  Reward Sharing}
\label{subsec:existence-friendship-RS}

In Section~\ref{subsec:existence-stable-matching} we showed that for
the case of equal sharing with friendship utilities, a stable matching
always exists. We showed this by proving that for equal sharing, a
stable matching without friendship utilities
(i.e. $\vec{\alpha}=\mathbf{0}$) is also a stable matching when we
have friendship utilities.

However, for unequal reward sharing with friendship, the set of stable
matchings for $\vec{\alpha}=\mathbf{0}$ is no longer a subset of the
set of stable matchings when we have friendship utilities.  Moreover,
existence of a stable matching for $\vec{\alpha}= \mathbf{0}$ no more
guarantees the existence of a stable matching with friendship
utilities. We will give examples below to justify both claims. Finally
we will conclude this section by giving a sufficient condition for the
existence of a stable matching for stable matching games with unequal
reward sharing and friendship utilities.


The following is an example which has non-overlapping sets of stable
matchings with and without friendship: Assign $r^u_{uw}=r^w_{uw}=1$,
$r^u_{uv}=10/11$, $r^v_{uv}=100/11$ with $\alpha_1=1/2$ and
$\alpha_2=\alpha_3=\cdots=0$ in Fig.~\ref{fig:swivel}. Without
friendship utilities, $\{(uw)\}$ is the only stable matching as node
$u$ will always want to get matched to node $w$. However, with
friendship utilities we have $q^u_{uv} = \frac{60}{11}$, $q^u_{uw} =
\frac{3}{2}, q^v_{uv} = \frac{105}{11}$, $q^v_{uw}=\frac{3}{2}$. Thus
using inequalities~\eqref{eqn:swivel-RS} and~\eqref{eqn:swivel-RS-2}
we see that with friendship utilities, the only stable matching is
$\{(uv)\}$ as the node $u$ will always want to get matched with node
$v$. Thus for unequal reward sharing with friendship utilties, the set
of stable matchings can be completely nonoverlapping with the set of
stable matchings for unequal reward sharing but without friendship
utilities.

\begin{figure}
  \centering
  \includegraphics[scale=0.5]{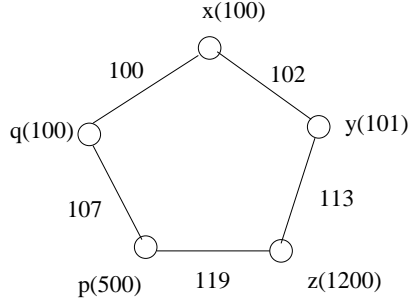}
  \caption{Existence of a stable matching without friendship
	does not guarantee existence of a stable matching
	with friendship}
  \label{fig:no-stable-matching-friendship}
\end{figure}

Now we give an example where we have a stable matching with
$\vec{\alpha}=\mathbf{0}$ but no stable matching with friendship
utilities. Consider the Matthew Effect sharing example as shown in
Fig.~\ref{fig:no-stable-matching-friendship}. The values on edges are
edge rewards of those edges. The values in the brackets beside a node
label is the brand valuee ($\lambda$ value) of that node. By
Theorem~\ref{thm:cycle-MEM-TM}, for $\vec{\alpha}=\mathbf{0}$ a stable
matching always exists for Matthew Effect sharing. However, let us
investigate the example for values shown in
Fig.~\ref{fig:no-stable-matching-friendship} with $\alpha_1=4/5,
\alpha_2=\alpha_3=\cdots=0$. Here we have
\begin{eqnarray*}
  q^q_{qx} = 90        &>& q^q_{pq} = 89.1667       \\
  q^x_{xy} = 91.7493   &>& q^x_{qx} = 90	    \\
  q^y_{yz} = 92.1545   &>& q^y_{xy} = 91.8507       \\
  q^z_{zp} = 112       &>& q^z_{yz} = 111.2455      \\
  q^p_{pq} = 103.4333  &>& q^p_{zp} = 102.2
\end{eqnarray*}
Suppose there exists a stable matching for this example. In such a
matching exactly one node would stay unmatched. Say the matching
$\{(qx),(zp)\}$ is a candidate for stable matching.  Now the node $y$
is unmatched. In such a situation, $(x,y)$ will form blocking pair
because we have $q^x_{xy} > q^x_{qx}$ and $q^y_{xy} > \alpha_1
r^x_{qx}$ (See inequalities~\eqref{eqn:swivel-RS}
and~\eqref{eqn:swivel-RS-2} substituting $\alpha_2=0$ as we use it in
this example). Hence $\{(qx),(zp)\}$ is not a stable
matching. Similarly every other matching can be shown to be not
stable. Hence here we do not have a stable matching with friendship
utilities, even though with $\vec{\alpha}=0$ a stable matching exists.

Now we give a sufficient condition for the existence of a stable
matching in unequal reward sharing with friendship utilities. Let us
denote by $SRP_q$ the instance of stable roommate problem where we
have exactly the same edges in the graph as our network but in
$SRP_q$ the nodes will prepare their preference lists based on
$q^x_{xy}$, i.e. a node $u$ will prefer node $v$ as roommate
over $w$ iff $q^u_{uv} > q^u_{uw}$, breaking ties arbitrarily.
Note that in an instance of stable roommate problem like $SRP_q$
friendship utilities plays no role.

\begin{theorem} \label{thm:existence-friendship-RS}
  A stable matching for $SRP_q$ is a stable matching for matching
  games with unequal reward sharing and friendship utilities.
  Hence, existence of a stable matching for $SRP_q$ implies the
  existence of a stable matching for general reward sharing with
  friendship utilities.
\end{theorem}

\begin{proof}
Suppose a stable matching $M$ for $SRP_q$ is not a stable matching
for unequal reward sharing with friendship
utilities. Then there exists a blocking
pair $(u,v)$ with one of the following two possibilities:
\begin{itemize}
  \item	In $M$, both $u$ and $v$ are matched:
	Let $u$ and $v$ be matched with $w$ and $z$ respectively.
	In such case, for $(u,v)$ to be a blocking pair
	the inequalities~\eqref{eqn:biswivel-RS-1}
	and~\eqref{eqn:biswivel-RS-2} must hold true. These inequalities
	imply $q^u_{uv} > q^u_{uw}$ and $q^v_{uv} > q^v_{vz}$. But
	then $(u,v)$ would be a blocking pair in $SRP_q$.
	Hence $M$ could not have been stable in $SRP_q$.
  \item	In $M$, only one of the nodes $u$ and $v$ is matched:
	Say $u$ is matched with $w$ but $v$ is unmatched.
	(Both cannot be unmatched otherwise $M$ would not be
	stable in $SRP_q$). Then for $(u,v)$ to be a blocking pair
	inequalities~\eqref{eqn:swivel-RS} and~\eqref{eqn:swivel-RS-2}
	must hold true. But these inequalities imply
	$q^u_{uv}>q^u_{uw}$ and thus $(u,v)$ would be a blocking
	pair in $SRP_q$. Hence $M$ could not have been stable
	in $SRP_q$.
\end{itemize}
Either way we reach a contradiction. Hence $M$ must be stable
with unequal reward sharing and friendship utilities. Moreover,
the set of stable matchings in $SRP_q$ is a subset of the
set of stable matchings in unequal reward sharing with friendship
utilities.
\qed \end{proof}
\subsection{Price of Anarchy with Friendship and General Reward Sharing}
\label{subsec:poa-friendship-RS}

This section is about proving the following theorem:
\begin{theorem} \label{thm:poa-friendship-RS}
  If a stable matching exists for general reward sharing with
  friendship utilities, then price of anarchy is at most $1+Q$, where
  $Q = \max_{(uv)\in E(G)}\frac{q^u_{uv}}{q^v_{uv}}
  =\frac{R+\alpha_1}{1+\alpha_1 R}$, and this bound is tight.
\end{theorem}
\begin{proof}
Let $M^*$ be a socially optimum matching, i.e., a matching with
maximum $\sum_{(uv)\in M^*} r_{uv}$. Let $M$ be any stable
matching. We will use $w^*_u$ (or $w_u$) to denote $q^u_{uv}$
if $u$ is matched to $v$ in $M^*$ (or in $M$).
Because $q_{xy}=(1+\alpha_1)r_{xy}$, we have
\begin{eqnarray}
  \text{ PoA  } = \max\limits_{\text{M is stable}}
	\cfrac{\sum_{(uv)\in M^*} q_{uv}}
	{\sum_{(uv)\in M}q_{uv}}
\end{eqnarray}
Using the definitions of $w^*_u$ and $w_u$ and letting
$w^*_u=0$ (or $w_u=0$) in case $u$ is unmatched in $M^*$ (or $M$),
we get
\begin{eqnarray}
  \text{ PoA  } = \max\limits_{\text{M is stable}}
	\cfrac{\sum_{u\in G} w^*_u}
	{\sum_{u\in G} w_u} \label{eqn:support-29}
\end{eqnarray}
If edge $(uv)\in M^*\setminus M$, then the utility of at least one
node among $u$ and $v$ does not increase if they were to deviate
in $M$ to get matched with each other. Say the utility of node $u$
does not increase. Now we have three cases:
\begin{itemize}
  \item	Both $u$ and $v$ are matched in $M$: Say
	$u$ and $v$ are matched to $w$ and $z$ respectively.
	In such a case if getting matched to $v$ does not
	increase the utility of $u$ then we have
	\begin{eqnarray}
	  q^u_{uv} \leq q^u_{uw} + \alpha_1 r^v_{vz}
		     + \alpha_{uz} r^z_{vz} \label{eqn:support-511}
	\end{eqnarray}
  \item $u$ is matched but $v$ is unmatched in $M$: Say
	$u$ is matched to $w$.
	In such a case if getting matched to $v$ does not
	increase the utility of $u$ then we have
	\begin{eqnarray}
	  q^u_{uv} \leq q^u_{uw} \label{eqn:support-512}
	\end{eqnarray}
  \item $u$ is unmatched but $v$ is matched in $M$: Say
	$v$ is matched to $z$.
	In such a case if getting matched to $v$ does not
	increase the utility of $u$ then we have
	\begin{eqnarray}
	  q^u_{uv} \leq \alpha_1 r^v_{vz}
			+\alpha_{uz} r^z_{vz} \label{eqn:support-513}
	\end{eqnarray}
\end{itemize}

Noticing that $\alpha_1 r^v_{vz} + \alpha_{uz} r^z_{vz} \leq r^v_{vz}
+ \alpha_1 r^z_{vz} = q^v_{vz}$, each of the
inequalities~\eqref{eqn:support-511}, \eqref{eqn:support-512},
and~\eqref{eqn:support-513} imply that:
\begin{eqnarray}
  q^u_{uv} \leq q^u_{uw} + q^v_{vz} \nonumber
\end{eqnarray}
A little algebraic manipulation gives us:
\begin{eqnarray}
  q^u_{uw} + q^v_{vz} &\geq& q^u_{uv} = \left(
		  \cfrac{1}{1+\frac{q^v_{uv}}{q^u_{uv}}} \right)
		  \cdot q_{uv} \nonumber \\
  \Rightarrow
  w_u + w_v &\geq& \cfrac{1}{1+Q}\cdot(w^*_u + w^*_v)
\end{eqnarray}
Adding such inequalities for all $(uv)\in M^*\setminus M$,
we obtain
\begin{eqnarray}
  \sum\limits_{(uv)\in M^*\setminus M} (w_u + w_v)
	  &\geq& \cfrac{1}{1+Q}\cdot
	         \sum\limits_{(uv)\in M^*\setminus M}
					(w^*_u + w^*_v)
			\label{eqn:support-40}
\end{eqnarray}
Notice that if a node $u$ appears in the above inequality
then $u$ is matched to different nodes in $M^*$ and $M$.
Denote the set of all such nodes by $B$. Hence
inequality~\eqref{eqn:support-40} becomes
\begin{eqnarray}
  \sum\limits_{u\in B} w_u &\geq& \cfrac{1}{1+Q} \cdot
			\sum\limits_{u\in B} w^*_u,
			\label{eqn:support-50}
\end{eqnarray}
and so the price of anarchy is at most $1+Q$, as desired.


\noindent
\textbf{Tightness of the bound:} Consider the $3$-length path as shown
in Fig.~\ref{fig:biswivel}. Make
$\alpha_2=\alpha_3=\cdots=0$. Substitute the following values:
\begin{align}
  r^u_{uv} = \frac{1}{1+\alpha_1}   & &r^v_{uv} = \frac{1}{1+\alpha_1}
						\nonumber\\
  r^u_{uw} = \frac{1}{1+\alpha_1 R} & &r^w_{uw} = \frac{R}{1+\alpha_1 R}
						\nonumber\\
  r^v_{vz} = \frac{1}{1+\alpha_1 R} & &r^z_{vz} = \frac{R}{1+\alpha_1 R}
						\nonumber
\end{align}
Note that as desired, $\max_{(xy)\in E(G)} \frac{r^x_{xy}}{r^y_{xy}} =
R$. Using $q^x_{xy} = r^x_{xy} + \alpha r^y_{xy}$, we obtain
\begin{align}
  q^u_{uv} = 1 & & q^v_{uv} = 1 \nonumber\\
  q^u_{uw} = 1 & & q^w_{uw} = Q \nonumber\\
  q^v_{vz} = 1 & & q^z_{vz} = Q \nonumber
\end{align}
Note that as desired, $\max_{(xy)\in E(G)} \frac{q^x_{xy}}{q^y_{xy}} =
\frac{R+\alpha_1}{1+\alpha_1 R} = Q$.  We have $\{(uv)\}$ as a stable
matching because given this matching, $(u,w)$ is not a blocking pair
as we have $q^u_{uw} \leq q^u_{uv}$. Similarly $(v,z)$ too is not a
blocking pair in matching $\{(uv)\}$.  Another stable matching is
$\{(uw),(vz)\}$ because given this matching, $(u,v)$ will not be a
blocking pair as we have $q^u_{uv} < q^u_{uw} +\alpha_1 r^v_{vz}$,
hence the condition in inequality~\eqref{eqn:biswivel-RS-1} is
violated.  Since there are no other stable matchings for this graph,
the price of anarchy will be determined by the value of the worst
stable matching which is $\{(uv)\}$. It is given by
\begin{eqnarray*}
  \text{ PoA  } = \frac{r_{uw}+r_{vz}}{r_{uv}}
		= \frac{q_{uw}+q_{vz}}{q_{uv}}
                = 1+Q
\end{eqnarray*}
Hence the bound is tight.
\qed \end{proof}

\paragraph{Discussion} We have $\text{PoA} \leq 1+Q$ where
$Q=\frac{R+\alpha_1}{1+\alpha_1 R}$. Let us consider the implications
of this bound. If $\alpha_1=0$, we have $\text{PoA} \leq 1+R$ which
agrees with Theorem~\ref{thm:poa-RS}. If $R=1$, we have
$\text{PoA}=2$. This result implies Theorem~\ref{thm:poa}, since when
we have $R=1$, then both $u$ and $v$, if they are matched to
each other, get the same reward from $(uv)$.

Notice also that $\frac{R+\alpha_1}{1+\alpha_1 R}$ is a decreasing
function of $\alpha_1$. As $\alpha_1$ goes from $0$ to $1$, the bound
goes from $1+R$ to $2$. Without friendship utilities, we have a tight
bound $\text{PoA}\leq 1+R$. Thus for $\vec{\alpha}=\mathbf{0}$, it can
be extremely bad if $R$ is large. As $\alpha_1$ gets close to $1$,
however, no matter how large $R$ is, PoA comes down to $2$ from
$R+1$. For example, if $\alpha_1=1/2$, then it is only 3. Thus, social
context can drastically improve the outcome for the society,
especially in the case of unfair and unequal reward sharing.

\subsection{Price of Stability with Friendship and General Reward
  Sharing}
\label{sec:pos-friendship-RS}

In this section, we give a simple lower bound $Q^\prime$ on the price
of stability for stable matching games with friendship and reward
sharing. Furthermore, we show that this bound is within an additive
factor of 1 of optimum, i.e., $Q< Q^\prime \leq \text{PoS} \leq 1+Q$.

To prove the lower bound, consider the 3-length path
as shown in Fig.~\ref{fig:biswivel}.
Make $\alpha_2=\alpha_3=\cdots=0$ and and use
the following values:
\begin{eqnarray*}
  r^u_{uv} = \frac{1}{1+\alpha_1}\left(
     \frac{1+\alpha_1 (R+1)}{(1+\alpha_1 R)} + \epsilon \right)& &
  r^v_{uv} = \frac{1}{1+\alpha_1}\left(
     \frac{1+\alpha_1 (R+1)}{(1+\alpha_1 R)} + \epsilon \right)
						\nonumber\\
  r^u_{uw} = \frac{1}{1+\alpha_1 R} & &
  r^w_{uw} = \frac{R}{1+\alpha_1 R}		\nonumber\\
  r^v_{vz} = \frac{1}{1+\alpha_1 R} & &
  r^z_{vz} = \frac{R}{1+\alpha_1 R}
\end{eqnarray*}
As desired we have $\max_{(xy)\in E(G)} \frac{r^x_{xy}}{r^y_{xy}}
= R$. Using $q^x_{xy} = r^x_{xy} + \alpha_1 r^y_{xy}$, we obtain
\begin{align}
  q^u_{uv} &= \frac{1+\alpha_1(R+1)}{1+\alpha_1 R} & &
  q^v_{uv} = \frac{1+\alpha_1(R+1)}{1+\alpha_1 R} \nonumber\\
  q^u_{uw} &= 1  & &
  q^w_{uw} = \frac{R+\alpha_1}{1+\alpha_1 R} \nonumber\\
  q^v_{vz} &= 1  & &
  q^z_{vz} = \frac{R+\alpha_1}{1+\alpha_1 R} \nonumber
\end{align}
As desired, we have $\max_{(xy)\in E(G)} \frac{q^x_{xy}}{q^y_{xy}} =
\frac{R+\alpha_1}{1+\alpha_1 R} = Q$.  We have $\{(uv)\}$ as a stable
matching because $(u,w)$ is not a blocking pair as $q^u_{uw} \leq
q^u_{uv}$.  Similarly $(v,z)$ will not be a blocking pair.  But the
matching $\{(uw),(vz)\}$ is no longer stable because $(u,v)$ is a
blocking pair as inequalities~\eqref{eqn:biswivel-RS-1}
and~\eqref{eqn:biswivel-RS-2} are satisfied. However $\{(uw),(vz)\}$
is still the socially optimal matching.  Hence the price of stability
for this graph will be given by
\begin{eqnarray*}
  \text{PoS} = \frac{r_{uw}+r_{vz}}{r_{uv}}
	     = \frac{q_{uw}+q_{vz}}{q_{uv}}
	     = \frac{(1+\alpha_1)(1+R)}{1+\alpha_1(R+1)}
\end{eqnarray*}
Let us define $Q^\prime = \frac{(1+\alpha_1)(1+R)}{1+\alpha_1(R+1)}$.
Because in the above instance we have $\text{PoS}=Q^\prime$, the lower
bound on the price of stability can be expressed as $\text{PoS}\geq
Q^\prime$, where $Q \le Q^\prime \le Q+1$. Since $Q+1$ is an upper
bound on the price of stability, this means that the lower bound of
$Q^\prime$ is within an additive term of 1 of optimum.

\begin{theorem} \label{thm:pos-friendship-RS}
The worst-case price of stability of stable matching games with friendship and
general reward sharing is in $[Q^\prime,Q+1]$, with $Q < Q^\prime \leq Q+1$.
\end{theorem}

\begin{proof}
The only part that is yet to be proven is $Q\leq Q^\prime$ and
$Q^\prime \leq 1+Q$. We have
\begin{eqnarray*}
 Q^\prime - Q
 = \frac{(1-\alpha_1+\alpha_1 R)(1+\alpha_1)}
	{(1+\alpha_1+\alpha_1 R)(1+\alpha_1 R)}
\end{eqnarray*}
As $(1-\alpha_1+\alpha_1 R) \leq (1+\alpha_1+\alpha_1 R)$ and
$1+\alpha_1 \leq 1+\alpha_1 R$, we have that $Q^\prime - Q \leq 1$.
As $R\geq 1$, the numerator is always positive.
Hence $0 < Q^\prime - Q \leq 1$. Using this with
$Q^\prime \leq \text{PoS}$,
we have that $Q < Q^\prime \leq \text{PoS} \leq 1+Q$.
\qed \end{proof}

%% file: ccg-1.tex
\section{Convex Contribution Games (CCGs)}
\label{sec:CCG}
In this section we consider convex contribution games (CCGs), as
defined in Section~\ref{sec:model}. In this version of CCG, players do
not have to spend all their budget: the total contribution of a player
to its incident edges must be {\em at most} $B_v$. This corresponds to
the fact that players may decide to keep some budget for themselves,
instead of spending it all on friendships/projects that the links
represent. We consider the case when players must spend their entire
budget in Section~\ref{sec:CCG2}.

For each CCG we define a \emph{corresponding stable matching game}
denoted $\sm$ as follows. The edge rewards in stable matching are
rewards when both players invest their full budget on an edge in the
CCG. For equal sharing this means $r_{uv} = f_{uv}(B_u,B_v)$, for
unequal sharing $r^u_{uv} = f^u_{uv}(B_u,B_v)$. For games with
friendship we assume the same values for $0 \le \alpha_1 \le \ldots
\le \alpha_{diam(G)} \le 1$ in both games. For simplicity, we use the
following notation: $g^u_{uv}(x,y) = f^u_{uv}(x,y) + \alpha_1
f^v_{uv}(x,y)$ and $g_{uv}(x,y) = f^u_{uv}(x,y) + f^v_{uv}(x,y)$ for
all $x, y \ge 0$.

In general, we will show that properties like existence and total
reward of pairwise equilibria in CCGs can be derived from the
properties of stable matchings in the corresponding games.

\subsection{Existence of a Pairwise Equilibrium}
\label{subsec:existence-pe-CCG}

We start by showing a general reduction for existence of a pairwise
equilibrium for arbitrary $\vec{\alpha}$. Recall that all reward functions of CCG are assumed to be convex in both its parameters, and satisfy the property that $f(x,0)=f(y,0)=0$ for all $x,y$. We call the class of such functions $C_0$. 

\begin{theorem}
  \label{thm:existence-pe-CCG}
  If all reward functions $f^u_{uv}(\cdot,\cdot)\in C_0$, then
  for every stable matching of the corresponding $\sm$ there is an
  equivalent pairwise equilibrium in the CCG. The pairwise equilibrium
  has the same assignment structure and total reward.
\end{theorem}

\begin{proof}
  Let $M$ be a stable matching in $\sm$ and consider the following
  strategy profile for the CCG: if node $u$ is matched to node $v$ in
  $M$, set $s_u(uv) = B_u$. If $u$ is not matched in $M$, then set
  $s_u(uv) = 0$ for all incident edges $(uv) \in E$. We will show
  that $s$ is a pairwise equilibrium. Obviously, $s$ has the same
  structure as $M$ and, in particular, yields the same total reward.

  First, note that $f^u_{uv}(x,y)$ is increasing and convex in
  \emph{both} arguments, which implies the same for
  $g^u_{uv}(x,y)$. Second, note that in $s$ for each edge we have both
  players contributing the full budget or nothing. Thus, players can
  deviate unilaterally or bilaterally only by reallocating budget onto
  edges $(uv) \not\in M$.

  Suppose two players $u$ and $v$ deviate and do not move any
  additional effort to their common edge $(uv)$ (because, e.g., $(uv)
  \not\in E$, or $(uv) \in M$ and both already spend all budget
  there). They cannot increase reward on incident edges in $M$ (if
  any), because they are spending their full budget. For every other
  incident edge $e \not\in M$, $e \neq (uv)$ they cannot increase the
  reward beyond 0, because the other player keeps putting 0
  effort. Hence, the only possibility to strictly improve their reward
  is when both players move some non-zero effort to $(uv)$. This, in
  particular, shows that there are no improving unilateral
  deviations.

  Hence, let us focus on bilateral deviations of players $u$ and $v$
  by moving some effort to a common edge. If both players are
  unmatched in $M$ and have such a improving deviation, this
  contradicts that $M$ is a stable matching. Hence, the following two
  cases remain.

  \begin{itemize}
  \item Suppose there exists a improving bilateral deviation onto
    $(uv) \not\in M$ and exactly one player, say $v$, is unmatched in
    $M$. Let $u$ be matched to $w$ in $M$. We assume that in the CCG
    $u$ and $v$ can improve by moving $\epsilon_1$ and $\epsilon_2$ of
    budget to $(uv)$, respectively. This implies
    \begin{eqnarray*}
      g^u_{uw}(B_u, B_w) &<& g^u_{uv}(\epsilon_1,\epsilon_2) + g^u_{uw}(B_u-\epsilon_1,B_w)
    \end{eqnarray*}
    As both $g^u_{uv}$ and $g^u_{uw}$ are convex in both arguments,
    this means that
    \begin{eqnarray*}
      g^u_{uw}(B_u,B_w) &<& g^u_{uv}(B_u,\epsilon_2) \quad < \quad g^u_{uv}(B_u,B_v)\enspace.
    \end{eqnarray*}
    For $SM(G,\vec{\alpha})$ this shows $g^u_{uv} > g^u_{uw}$, but
    then $v$ cannot be unmatched, because this would contradict that
    $M$ is a stable matching.

  \item Suppose there exists a proftiable bilateral deviation onto an
    edge $(uv) \not\in M$, where $u$ is matched to $w$ and $v$ to $z$
    in $M$. If $u$ and $v$ transfer $\epsilon_1$ and $\epsilon_2$ to
    $(uv)$, respectively, then for node $u$ we have
    \begin{eqnarray}
      g^u_{uv}(\ep_1,\ep_2)
      &>& g^u_{uw}(B_u,B_w)-g^u_{uw}(B_u-\ep_1,B_w) \nonumber\\
      & & + \alpha_1 (f^v_{vz}(B_v,B_z) - f^v_{vz}(B_v-\ep_2,B_z)) \nonumber\\
      & & + \alpha_{uz} (f^z_{vz}(B_v,B_z) - f^z_{vz}(B_v-\ep_2,B_z))\enspace,
      \label{eqn:support-221}
    \end{eqnarray}
    which formally states that there is a net increase in the utility
    of $u$ because of the transfer. Similarly for node $v$ we have
    \begin{eqnarray}
      g^v_{uv}(\ep_1,\ep_2)
      &>& g^v_{vz}(B_v,B_z)-g^v_{vz}(B_v-\ep_2,B_z) \nonumber\\
      & & + \alpha_1 (f^u_{uw}(B_u,B_w) - f^u_{uw}(B_u-\ep_1,B_w)) \nonumber\\
      & & + \alpha_{vw} (f^w_{uw}(B_u,B_w) - f^w_{uw}(B_u-\ep_1,B_w))\enspace.
      \label{eqn:support-223}
    \end{eqnarray}
    As all functions $f$ and $g$ are convex and increasing in both
    arguments, we get
    \begin{eqnarray}
      g^u_{uv}(B_u,B_v) &>& g^u_{uw}(B_u,B_w)
      + \alpha_1 f^v_{vz}(B_v,B_z)
      + \alpha_{uz} f^z_{vz}(B_v,B_z) \nonumber\\
      g^v_{uv}(B_u,B_v) &>& g^v_{vz}(B_v,B_z)
      + \alpha_1 f^u_{uw}(B_u,B_w)
      + \alpha_{vw} f^w_{uw}(B_u,B_w)\enspace, \nonumber
    \end{eqnarray}
    but then in $\sm$ the following must hold true
    \begin{eqnarray}
      q^u_{uv}&>& q^u_{uw}
      + \alpha_1 r^v_{vz}
      + \alpha_{uz} r^z_{vz}\nonumber\\
      q^v_{uv}&>& q^v_{vz}
      + \alpha_1 r^u_{uw}
      + \alpha_{vw} r^w_{uw}\enspace.\nonumber
    \end{eqnarray}
    This means that in $\sm$, nodes $u$ and $v$ would prefer getting
    matched to each other (see
    inequalities~\eqref{eqn:biswivel-RS-1}
    and~\eqref{eqn:biswivel-RS-2})), i.e., $M$ is not a
    stable matching in $\sm$ which contradicts our assumption.
  \end{itemize}
\qed \end{proof}

The conditions for existence of a pairwise equilibrium can be weakened
for CCGs without friendship. In this case, convexity of reward share
in the other player's contribution is not necessary.

\begin{corollary}\label{cor:existence-pe-CCG-no-friendship}
  If $f^u_{uv}(s_u(uv),s_v(uv))$ are increasing in $s_u(uv)$ and
  $s_v(uv)$ and convex in $s_u(uv)$, then for every stable matching of
  the corresponding game $SM(G,\vec{\alpha}=\mathbf{0})$, there is an
  equivalent pairwise equilibrium in the CCG without friendship. The
  pairwise equilibrium has the same assignment structure and total
  reward.
\end{corollary}

\begin{proof}
  We just need to observe that $f^u_{uv}(s_u(uv),s_v(uv))$ does not
  need to be convex in $s_v(uv)$. In the proof of the previous
  theorem, convexity in $s_v(uv)$ is only required in
  inequality~\eqref{eqn:support-221} and~\eqref{eqn:support-223} with
  the $\alpha_2$ coefficient.
\qed \end{proof}

We proceed to specify more detailed results for particular reward
sharing rules.

\paragraph{Equal Sharing}
We first consider equal sharing with (or without) friendship. In this
case, $\sm$ always has a stable matching from
Theorem~\ref{thm:existence-stable-matching}. Also, because
$f_{uv}(s_u(uv),s_v(uv))=f^u_{uv}(s_u(uv),s_v(uv))=
f^v_{uv}(s_u(uv),s_v(uv))$ and $f_{uv}(\cdot,\cdot)\in C_0$, we also
have $f^u_{uv}(\cdot,\cdot)\in C_0$ and $f^u_{uv}(\cdot,\cdot)\in
C_0$. Hence all the conditions for existence of a pairwise equilibrium
are satisfied. The following corollary extends a main result
from~\cite{AnshelevichAlgo11} to CCGs with arbitrary friendship.
\begin{corollary}\label{cor:existence-pe-CCG-equal-sharing}
  A pairwise equilibrium always exists in a CCG with equal sharing.
\end{corollary}

\paragraph{Matthew Effect Sharing}
Next, let us consider Matthew Effect CCGs defined as follows: Each
node $u\in G$ has an associated brand value $\lambda_u$. If nodes $u$
and $v$ invest $s_u(uv)$ and $s_v(uv)$ respectively on edge $(uv)$,
then $u$ obtains a reward of
\[ f^u_{uv}(s_u(uv),s_v(uv)) =
\frac{\lambda_u}{\lambda_u+\lambda_v}f_{uv}(s_u(uv),s_v(uv))\]
from edge $(uv)$. Consequently in Matthew Effect CCG, we have
\[ g^u_{uv}(s_u(uv),s_v(uv)) = \frac{\lambda_u+\alpha_1 \lambda_v}
{\lambda_u+\lambda_v}f_{uv}(s_u(uv),s_v(uv))\enspace.\]

It can be easily seen that in Matthew Effect CCG,
$f^u_{uv}(\cdot,\cdot)$ are increasing and convex in the investment of
$u$ and $v$. Hence we have the following corollaries from
Theorem~\ref{thm:existence-pe-CCG} and Lemma~\ref{thm:cycle-MEM-TM}.
\begin{corollary}
  For every stable matching of the corresponding game $\sm$ with
  Matthew Effect Sharing, there is an equivalent pairwise equilibrium
  in the Matthew Effect CCG. The pairwise equilibrium has the same
  assignment structure and total reward.
\end{corollary}

As a special case, we have guaranteed existence for Matthew Effect
CCGs  without friendship.

\begin{corollary}
  A pairwise equilibrium always exists in Matthew Effect CCGs without
  friendship.
\end{corollary}

\paragraph{Proportional Sharing}
Finally, let us consider a natural model of sharing that is specific
to CCGs (this model was not considered in Section \ref{sec:reward-sharing}). In Proportional Sharing CCG, the reward a node gets is proportional to
the effort it contributes to an edge. In other words, if nodes $u$ and
$v$ invest $s_u(uv)$ and $s_v(uv)$ respectively on edge $(uv)$, then
$u$ gets a reward of
\[ f^u_{uv}(s_u(uv),s_v(uv)) =
\frac{s_u(uv)}{s_u(uv)+s_v(uv)}f_{uv}(s_u(uv),s_v(uv))\]
from edge $(uv)$. Consequently, in Proportional Sharing CCG, we have
\[ g^u_{uv}(s_u(uv),s_v(uv)) = \frac{s_u(uv)+\alpha_1 s_v(uv)}
{s_u(uv)+s_v(uv)}f_{uv}(s_u(uv),s_v(uv))\enspace.\]

For a proportional sharing CCGs, it can be verified that
$f^u_{uv}(s_u(uv),s_v(uv))$ are increasing in $s_u(uv)$ and $s_v(uv)$
and convex in $s_u(uv)$. It is easy to observe that the corresponding
stable matching game is, in fact, an instance of the Matthew Effect
model with $\lambda_u=B_u$ for all nodes $u\in G$. Hence, without
friendship a stable matching always exists in the stable matching game
and provides the following corollary based on
Corollary~\ref{cor:existence-pe-CCG-no-friendship}:
\begin{corollary}
  For every stable matching of the corresponding game
  $SM(G,\vec{\alpha}=\mathbf{0})$ with Matthew Effect Sharing, there
  is an equivalent pairwise equilibrium in the Proportional Sharing
  CCG without friendship. The pairwise equilibrium has the same
  assignment structure and total reward. There always exists at least
  one such stable matching with corresponding pairwise equilibrium.
\end{corollary}
\subsection{Prices of Anarchy and Stability}
\label{subsec:poa-pos-CCG}

In the previous section, we have seen that stable matchings can easily
be translated into pairwise equilibria for CCGs. However, there could
potentially be other pairwise equilibria that are, in particular, much
worse in terms of total reward. In this section, we show that this is
not the case and translate the bounds for prices of anarchy and
stability from stable matching to CCGs. Thus, these bounds apply for all reward functions $f^u_{uv}(\cdot,\cdot) \in C_0$. 

Let us define a \emph{tight edge} as an edge on which both of the
nodes invest their full budget. The social optimum with maximum total
reward does not depend on the reward sharing scheme at hand or on the
values of $\vec{\alpha}$. Thus, Claim 2.10 in~\cite{AnshelevichAlgo11}
shows that there always exists a \emph{tight social optimum}, i.e., a
social optimum $s^*$ such that players invest only in tight edges.
In particular, as the CCG allows more flexibility than the
corresponding stable matching game $\sm$, a tight social optimum in
the CCG is in 1-to-1 correspondence to a social optimum in
$\sm$. Whenever stable matchings in $\sm$ correspond to pairwise
equilibria in the CCG and a tight social optimum in the CCG
corresponds to a social optimum in the corresponding $\sm$, we can
directly translate our upper bounds on the price of stability to
CCGs. This implies the following corollary.

\begin{corollary}
  The price of stability in CCGs with equal sharing and friendship is
  at most $\frac{2+2\alpha_1}{1+2\alpha_1+\alpha_2}$. \BRBP\ starting
  from a tight social optimum converges in polynomial time to a
  pairwise equilibrium that achieves this bound.
\end{corollary}

For the price of anarchy, we could possibly have worse equilibria in
the CCG that do not correspond to matchings in $\sm$. However, the
same bound as in Theorem~\ref{thm:poa-friendship-RS} can be proved. We
use parameter $Q$ (as detailed in Section~\ref{sec:friendship-RS}) for
the corresponding stable matching game $\sm$.

\begin{theorem}\label{thm:poa-CCG}
  The price of anarchy in CCGs is bounded by $\text{PoA}\leq Q+1$. 
\end{theorem}

\begin{proof}
  It suffices to compare against a tight social optimum which we
  denote by $s^*$. Let $s$ denote a pairwise equilibrium. Let
  $s_u(uv)$ denote the investment of node $u$ on edge $(uv)$ in the
  pairwise equilibrium $s$. Let us define $w_u$ as:
  \begin{eqnarray}
    w_u = \sum\limits_{z\in N_1(u)} g^u_{uz}(s_u(uz),s_z(uz))
  \end{eqnarray}
  Using the definition of $w_u$, it can be verified that the total
  reward $w(s)$ of $s$ is given by
  \begin{eqnarray}
    w(s) = \frac{1}{1+\alpha_1}\sum_{u\in G}w_u
  \end{eqnarray}
  and the total reward $w(s^*)$ of $s^*$ is
  \begin{eqnarray}
    w(s^*)=\frac{1}{1+\alpha_1}\sum_{(uv)\in s^*}g_{uv}(B_u,B_v)
  \end{eqnarray}
  where by $(uv)\in s^*$ we mean the tight edges in $s^*$.  Hence the
  price of anarchy can also be expressed as
  \begin{eqnarray}
    \text{PoA  } = \max_{s}\frac{w(s^*)}{w(s)}
    = \max_{s}\frac{\sum_{(uv)\in s^*}g_{uv}(B_u,B_v)}
    {\sum_{u\in G}w_u}
    \label{eqn:support-230}
  \end{eqnarray}
  We will use this alternative expression for the price of anarchy for
  proving the claim.

  Let us construct a set of \textit{witness} nodes in one-to-one
  correspondence with tight edges in $s^*$ as follows: For each tight
  edge $(uv)$ of $s$ we make either $u$ or $v$ a witness for $(uv)$. On other
  edges, we have either $s_u(uv) < B_u$ or $s_v(uv) < B_v$. Now as
  $s$ is a pairwise equilibrium, if $u$ and $v$ both transfer their
  full budget to $(uv)$, utility of at least one -- we w.l.o.g.\
  assume node $u$ -- will not increase, and we make $u$ witness for
  $(uv)$.

  As the deviation towards $(uv)$ is not improving, we can examine
  the utility $u$ and bound
  \begin{eqnarray*}
    g^u_{uv}(B_u,B_v)
    &\leq&
    \sum\limits_{y\in N_1(u)} g^u_{uy}(s_u(uy), s_y(uy))
    + \alpha_1 \sum\limits_{z\in N_1(v)-u}
    f^v_{vz}(s_v(v), s_z(vz))
    \nonumber \\
    & &+ \sum\limits_{z\in N_1(v)-u} \alpha_{uz} f^z_{vz}(s_v(vz), s_z(vz))
  \end{eqnarray*}

  As a consequence of $\alpha_1 f^v_{vz}(s_v(vz), s_z(vz))+ \alpha_{uz}
  f^z_{vz}(s_v(vz), s_z(vz)) \leq \alpha_1 f^v_{vz}(s_v(vz),
  s_z(vz))+ \alpha_1 f^z_{vz}(s_v(vz), s_z(vz)) \leq
  f^v_{vz}(s_v(vz), s_z(vz))+ \alpha_1 f^z_{vz}(s_v(vz), s_z(vz))
  = g^v_{vz}(s_v(vz),s_z(vz))$, we obtain

  \begin{eqnarray}
    g^u_{uv}(B_u,B_v)
    &\leq&
    \sum\limits_{y\in N_1(u)} g^u_{uy}(s_u(uy), s_y(uy))
    + \sum\limits_{z\in N_1(v)-u} g^v_{vz}(s_v(vz), s_z(vz))
    \label{eqn:support-232}\\
    \Rightarrow
    \cfrac{1}{1+Q} \cdot g_{uv}(B_u,B_v)
    &\leq&
    \sum\limits_{y\in N_1(u)} g^u_{uy}(s_u(uy), s_y(uy))
    + \sum\limits_{z\in N_1(v)} g^v_{vz}(s_v(vz), s_z(vz))
    \\
    \Rightarrow
    \cfrac{1}{1+Q} \cdot g_{uv}(B_u,B_v)
    &\leq& w_u + w_v \label{eqn:support-233}
  \end{eqnarray}

  Note that the last inequality~\eqref{eqn:support-233} also holds for
  tight edges $(uv)$.  Thus we have one inequality due to witnessing
  each edge. Adding these inequalities,
  \begin{eqnarray*}
    \cfrac{1}{1+Q} \sum\limits_{(uv)\in s^*} g_{uv}(B_u,B_v)
    &\leq& \sum\limits_{(uv)\in s^*} (w_u + w_v)
    \leq \sum\limits_{u\in G} w_u
  \end{eqnarray*}
  Hence we get
  \begin{eqnarray*}
    \frac{\sum_{(uv)\in s^*}g_{uv}(B_u,B_v)} {\sum_{u\in G}w_u}
    \leq 1+Q
  \end{eqnarray*}
  As this is valid for any pairwise equilibrium $s$, using
  Eqn.~\eqref{eqn:support-230}, we can complete the proof by finding
  \begin{eqnarray*}
    \text{PoA  }\leq 1+Q\enspace.
  \end{eqnarray*}
\qed \end{proof}

%
%
%
%

%% file: ccg-2.tex
\section{Contribution Games With Tight Budget Constraints}\label{sec:CCG2}


In this section we consider the version of contribution games where
all player budget {\em must} be spent on adjacent edges, i.e., the sum
of each player $v$'s contributions to incident edges
$\sum_{(v,u)}s_v(vu)$ exactly equals $B_v$. At first glance, this
version of the game does not seem very different than the case when
$\sum_{(v,u)}s_v(vu)\leq B_v$. And in fact, when node utilities simply
consist of $R_u(s)$ (i.e., there is no ``friendship" component), then
it can be easily shown that all the results from
\cite{AnshelevichAlgo11} and from Section~\ref{sec:CCG} still hold.

The presence of friendship utilities, however, makes a large
difference. Consider, for instance, the simple example in
Fig.~\ref{fig:biswivel} with reward functions
$f_{uw}(x,y)=f_{vz}(x,y)=(1-\varepsilon)\min(x,y)$,
$f_{uv}(x,y)=\min(x,y)$, all node budgets equal to 1, and
$\alpha_1=1/2,~\alpha_2=\alpha_3=0$. If nodes were allowed to
contribute less than their budget, then the solution where the two
middle nodes put all their budget on the edge between them, and the
two endpoints do not contribute anything, is a pairwise
equilibrium. If, however, the two endpoint nodes {\em must} contribute
their budget to their incident edges, then this is no longer a
pairwise equilibrium, as the two middle nodes are able to
simultaneously move their budgets to the outer edges, obtaining
$(1+2\alpha_1)(1-\varepsilon)>(1+\alpha_1)$ utility each. In general,
the argument for existence of pairwise equilibrium from
Section~\ref{sec:CCG} no longer works, as stable matchings may no
longer correspond to pairwise equilibrium as they did in
Theorem~\ref{thm:existence-pe-CCG}. Moreover, the existence argument
from~\cite{AnshelevichAlgo11} is based on forming a maximal greedy
matching, which, as the above example shows, is not necessarily a
pairwise equilibrium.

Fortunately, we are able to extend many of our results to the version
where players must spend their entire budgets as well. Specifically,
we show that all our results still hold for the case of equal sharing,
with $\alpha_i=0$ for all $i\geq 2$. We call this type of perceived
utility {\em local friendship}, since nodes only care about their
neighbors, but not their neighbors-of-neighbors. In the rest of this
section, we let $\alpha=\alpha_1$, since it is the only non-zero
$\alpha_i$.

As the example above shows, even for the case of local friendship,
contribution games with tight budget constraints can behave
peculiarly. Essentially, the complication here arises from the fact
that a stable matching is stable with respect to swivel and biswivel
deviations only. On the other hand, a pairwise equilibrium has to be
stable with respect to all bilateral deviations, including two
non-adjacent nodes switching contributions, or two adjacent nodes
moving their contributions away from their shared edge. If all
unmatched nodes do not contribute anything to incident edges, then all
these deviations cannot be improving (see
Theorem~\ref{thm:existence-pe-CCG}), but in the model where all budget
must be spent, these deviations can and do occur.

The key to our results in this section is the following theorem. To
state the theorem, we first need the concept of forbidden edges,
defined below. As in Section~\ref{sec:CCG}, $\sm$ is the stable
matching game corresponding to a CCG, and $r_{uv} = f_{uv}(B_u,B_v)$
as before.

\begin{definition}
  We call an edge $e=(u,v)$ in a contribution game {\bf forbidden} if
  both of the following hold:
\begin{enumerate}
\item There exist edges $(u,x)$ and $(v,y)$ with $x\neq v$, $y\neq u$,
  and both $x,y$ having degree 1.

\item $u$ and $v$ would be willing to deviate by putting all their
  budget on edges $(u,x)$ and $(v,y)$ even if both $u$ and $v$ are
  putting all their budget on edge $e$. In other words, $e$ is such
  that $r_{uv}+\alpha r_{uv}<r_{ux}+\alpha r_{ux}+\alpha r_{vy}$ and
  $r_{uv}+\alpha r_{uv}<r_{vy}+\alpha r_{vy}+\alpha r_{ux}$.
\end{enumerate}
\end{definition}

\begin{theorem}\label{thm:matchingToEquilibrium=Budget}
  Consider a CCG with equal reward sharing and local friendship in
  which all players must contribute their entire budget to incident
  edges. Let a matching $M$ be a stable matching in the corresponding
  $\sm$. If $M$ does not contain any forbidden edges, then there
  exists an equivalent pairwise equilibrium, with the same assignment
  structure and total reward, obtained by having all unmatched nodes
  split their contributions equally among all incident edges.
\end{theorem}

\begin{proof}
  Set the strategy of a node $u$ with an edge $(u,v)\in M$ to put all
  its budget onto edge $(u,v)$. If a node is unmatched in $M$, then we
  set its strategy to spread its budget equally among all its incident
  edges. Call this solution $s$; our goal is to show that $s$ is a
  pairwise equilibrium. Since all functions $f_e$ are nondecreasing
  and convex in both parameters, we can restrict our attention, wlog,
  to deviations where nodes move all their budget to a single edge.

  It is clear that no unilateral improving deviations exist in
  $s$. This is because if a node $v$ could gain utility by
  unilaterally moving effort to an edge $e=(u,v)$, then $u$ must be
  unmatched (recall that $f_e(B_v,0)=0$ by definition of CCG), and so
  $v$ would gain by performing a swivel to $u$ in $M$, contradicting
  the stability of $M$. Similarly, swivels and biswivels (i.e.,
  deviations where two nodes $u$ and $v$ move all their effort to the
  edge $(u,v)$) cannot be improving deviations, since otherwise $M$
  would not be a stable matching. The above argument also implies that
  no unmatched node would participate in an improving deviation from
  $s$, since it can only obtain positive reward by putting effort on
  an edge to another unmatched node, which contradicts the stability
  of $M$, since all stable matchings are maximal.

  Now we must consider all other types of deviations. Specifically, we
  must now show that for every pair of nodes $u$ and $v$ that are
  matched in $M$ (not necessarily to each other), there is no
  improving bilateral deviation of $u$ and $v$. Let $e_1=(u,w)$ and
  $e_2=(v,z)$ be the edges of $M$ incident to $u$ and $v$
  respectively. Note that $e_1$ may equal $e_2$.

  Suppose to the contrary that $u$ and $v$ have an improving bilateral
  deviation in $s$, and let $e_3=(u,x)$ be the edge that $u$ moves its
  budget to, and $e_4=(v,y)$ be the edge that $v$ moves its budget
  to. We know that $e_3\neq e_4$, since otherwise this deviation
  corresponds to an improving biswivel in $M$, which is not possible
  since $M$ is a stable matching. This means that wlog, in a bilateral
  deviation $u$ and $v$ will move all their budget to an edge incident
  to an unmatched node: this is because by moving its budget to an
  edge incident to a matched node that is neither $e_1$ or $(u,v)$,
  $u$ will end up with 0 reward. Thus, if an improving bilateral
  deviation of $u$ and $v$ exists, then in this deviation $u$ and $v$
  move their budgets from $e_1$ and $e_2$ (which may be the same edge
  $e_1=(u,v)=e_2$) to edges $e_3=(u,x)$ and $e_4=(v,y)$, which are
  {\em not} the same edges. It is still possible, however, that $x$
  may equal $y$.

  Denote by $\gamma_x$ the reward on edge $e_3$ obtained if node $u$
  puts all its budget onto this edge, i.e.,
  $\gamma_x=f_{ux}(B_u,s_x(ux))$. Note that if $x$ has only a single
  incident edge, then $\gamma_x=r_{ux}$, since $s_x(ux)=B_x$ in this
  case. If instead $x$ has degree at least 2, then $\gamma_x\leq
  r_{ux}/2$, since $x$ is splitting its budget evenly among all
  incident edges, and since $f_{ux}$ is convex in the contribution of
  node $x$. Similarly, define $\gamma_y$ as
  $f_{vy}(B_v,s_y(vy))$. Finally, we will use notation $[a]_P$ to
  denote $a$ if property $P$ holds, and 0 otherwise.

  \paragraph{Case 1: $e_1=e_2=(u,v)$} In this case, it cannot be that
  both $x$ and $y$ have degree 1, since this would imply that $(u,v)$
  is a forbidden edge, and thus could not be in $M$. Therefore, we can
  assume that, wlog, the degree of $x$ is at least 2. The only rewards
  that change during the deviation are the rewards on $(u,v)$, $e_3$,
  and $e_4$. The reward $u$ receives from edges $e_3$ and $e_4$ after
  the deviation is at most
  $\gamma_x+\alpha\gamma_x+\alpha\gamma_y+\alpha\gamma_y$; the last
  term is only present if node $y$ is adjacent to $u$. Since $x$ has
  degree at least 2, then in order for this to be a profitable
  deviation for $u$, it must be that

  \begin{equation}\label{eq.case2.1}
    \frac{r_{ux}+\alpha r_{ux}}{2}+\alpha r_{vy}> r_{uv}+\alpha r_{uv}.
  \end{equation}

  Recall that nodes $x$ and $y$ are unmatched in $M$. Due to stability
  of $M$, it must be that $r_{uv}\geq r_{ux}$ and $r_{uv}\geq r_{vy}$,
  since otherwise swiveling from $(u,v)$ to $(u,x)$ or from $(u,v)$ to
  $(v,y)$ would be an improving swivel deviation in the stable
  matching. Thus, Inequality~\eqref{eq.case2.1} implies that
  $(1+\alpha)r_{uv}<(\frac{1}{2}+\frac{3\alpha}{2})r_{uv}$, which is a
  contradiction since $\alpha\leq 1$.

  \paragraph{Case 2: $e_1\neq e_2$} We now have the final case to
  consider, in which $e_1\neq e_2$ (recall also that $e_3\neq
  e_4$). The total contribution of rewards of $e_1=(u,w)$ and
  $e_2=(v,z)$ to the utilities of $u$ and $v$ before the deviation was
  at least
  \begin{equation}\label{eq.case2.before1}
    (1+\alpha+[\alpha]_{(u,v)\in E})(r_{uw}+ r_{vz}).
  \end{equation}
  Note that the contribution can be even larger if, for example, $u$
  is adjacent to $z$, but it is at least as large
  as~\eqref{eq.case2.before1}. The total contribution of rewards of
  $e_3$ and $e_4$ to $u$ and $v$ after the deviation is at most
  \begin{equation}\label{eq.case2.after1}
    (1+\alpha+[\alpha]_{(u,v)\in E})(\gamma_x+\gamma_y)+[\alpha\gamma_x]_{deg(x)>1}+[\alpha\gamma_y]_{deg(y)>1},
  \end{equation}
  where $deg(x)$ is the degree of $x$. Thus for this deviation to be
  strictly improving, it must be that
  $\eqref{eq.case2.after1})>\eqref{eq.case2.before1})$.

  If $deg(x)=1$, then $(1+\alpha+[\alpha]_{(u,v)\in
    E})(\gamma_x)+[\alpha\gamma_x]_{deg(x)>1}=(1+\alpha+[\alpha]_{(u,v)\in
    E})r_{ux}$, since in this case $\gamma_x=r_{ux}$. If $deg(x)>1$,
  then $(1+\alpha+[\alpha]_{(u,v)\in
    E})(\gamma_x)+[\alpha\gamma_x]_{deg(x)>1} \leq
  (1+2\alpha+[\alpha]_{(u,v)\in E})(r_{ux}/2),$ since in this case
  $\gamma_x\leq r_{ux}/2,$ as argued above. Thus in either case,
  $(1+\alpha+[\alpha]_{(u,v)\in
    E})(\gamma_x)+[\alpha\gamma_x]_{deg(x)>1}\leq
  (1+\alpha+[\alpha]_{(u,v)\in E})r_{ux},$ and so quantity
  \eqref{eq.case2.after1} is at most
  \begin{equation}\label{eq.case2.after2}
    (1+\alpha+[\alpha]_{(u,v)\in E})(r_{ux}+ r_{vy})
  \end{equation}

  As argued above, since both $x$ and $y$ are unmatched in $M$, then
  we know that 
  $r_{uw}\geq r_{ux}$ and $r_{vz}\geq r_{vy}$, since otherwise $M$
  would have an improving swivel. Therefore,
  quantity~\eqref{eq.case2.before1} is at
  least~\eqref{eq.case2.after2}), and so $\eqref{eq.case2.after1} \leq
  \eqref{eq.case2.before1}$. Therefore, this cannot be an improving
  deviation.
\qed \end{proof}

Using the above theorem, we can proceed similarly to
Section~\ref{sec:CCG}, and show existence of pairwise equilibrium,
convergence results, and the same bounds for price of stability as
before. 

\begin{theorem}
  Consider a CCG with equal reward sharing and local friendship
  utilities in which all players must contribute their entire budget
  to incident edges. Then, the price of stability is at most
  $\frac{2+2\alpha}{1+2\alpha}$, and a pairwise equilibrium that
  achieves this bound exists and can be found in polynomial time.
\end{theorem}

\begin{proof}
  Recall that, by the discussion at the start of
  Section~\ref{subsec:poa-pos-CCG}, the optimal solution in this game
  simply corresponds to the maximum-weight matching $M^*$. First,
  notice that $M^*$ does not contain forbidden edges. This is easy to
  see, since if a forbidden edge $e=(u,v)\in M^*$, then the nodes $x$
  and $y$ of degree 1 adjacent to $u$ and $v$ must be
  unmatched. Removing $(u,v)$ from $M^*$, and adding $(u,x)$ and
  $(v,y)$, increases the weight of the matching. This is because, by
  definition of forbidden
  edge, $$r_{uv}<\frac{1+2\alpha}{2+2\alpha}(r_{ux}+r_{vy}),$$ and
  thus $r_{uv}<r_{ux}+r_{vy}$. Since $M^*$ is the maximum weight
  matching, this is a contradiction, and thus $M^*$ cannot contain
  forbidden edges. Moreover, every node adjacent to a forbidden edge
  must be matched in $M^*$, otherwise we could add the edge between
  this node and a node of degree 1 to increase the weight of $M^*$.

  Now consider the same game, but with all forbidden edges removed
  from the graph. The maximum-weight matching does not change, and
  thus the optimum solution does not change. Let a matching $M$ be the
  matching found by \BRBP\ in the corresponding matching game $\sm$
  (this is the game with all forbidden edges removed). Define a
  solution $s$ to be such that all nodes matched in $M$ put all their
  effort on edges of $M$, and all unmatched nodes split their effort
  equally among all incident edges. By
  Theorem~\ref{thm:pos-friendship-equal-sharing}, the weight of $M$ is
  at least $\frac{1+2\alpha}{2+2\alpha}$ of the weight of $M^*$, and
  thus $s$ meets the desired price of stability bound.  $s$ can
  clearly be found in poly-time, since $M$ can be found in poly-time.

  To show that $s$ is a pairwise equilibrium, we will use
  Theorem~\ref{thm:matchingToEquilibrium=Budget}. $M$ clearly does not
  contain forbidden edges, since all of these edges were removed
  before forming $M$. It is also stable with respect to deviations to
  non-forbidden edges, since \BRBP\ forms a stable matching. Thus, all
  we need to show is that $M$ is stable with respect to deviations to
  forbidden edges.

  To show this, consider a forbidden edge $(u,v)$, with nodes $x$ and
  $y$ defined as in the definition of forbidden edge. In matching $M$,
  $u$ must be matched to some node $w\neq v$: this is because it
  cannot be matched to $v$ (we removed edge $(u,v)$ when running our
  algorithm to form the matching), and it cannot be unmatched since
  nodes $u$ and $x$ would then form a blocking pair in $M$. Moreover,
  $r_{uw}\geq r_{ux}$, since otherwise $u$ would have an improving
  swivel to node $x$ in matching $M$. The same holds for node $v$: it
  must be matched to some node $z$ such that $r_{vz}\geq r_{vy}$. By
  definition of forbidden edge, this implies that $(u,v)$ does not
  form a blocking pair, even with edge $(u,v)$ present.

  We have now shown that $M$ is a stable matching that does not
  contain forbidden edges, and thus $s$ is a pairwise equilibrium, as
  desired.
\qed \end{proof}

Finally, since a pairwise equilibrium when nodes spend their entire
budget is also a pairwise equilibrium for the CCG without tight budget
constraints, then we obtain the following corollary of
Theorem~\ref{thm:poa-CCG}.

\begin{corollary}\label{cor:poa-CCG2}
  The price of anarchy in CCGs in which all players must contribute
  their entire budget to incident edges is bounded by $\text{PoA}\leq
  1+Q$.
\end{corollary}